\documentclass[12pt]{article}

\usepackage{graphicx}
\usepackage[dvips]{color}

\usepackage{amsmath,amssymb,epsf,epic}

\newcommand{\h}{{\mathfrak{h}}}

\newcommand{\un}{{\mathbb I}}
\newcommand{\ra}{\rightarrow}
\newcommand{\tr }{\mbox{tr}}
\renewcommand{\ker}{{\rm Ker~}}
\newcommand{\ran}{{\rm Ran~}}

\newcommand{\bra}{\langle} 

\newcommand{\ket}{\rangle}

\newcommand{\D}{{\mathbb D}}

\renewcommand{\vec}[1]{\mathbf{#1}}

\newcommand{\be}{\begin{equation}}
\newcommand{\ee}{\end{equation}}
\newcommand{\bea}{\begin{eqnarray}}
\newcommand{\eea}{\end{eqnarray}}

\newcommand{\eps}{\epsilon}

\newcommand{\ffi}{\varphi}

\newcommand{\ep}{\hfill  {\vrule height 10pt width 8pt depth 0pt}}

\newcommand{\grintl}{[\kern-.18em [}
\newcommand{\grintr}{]\kern-.18em ]}

\newcounter{resultcounter}[section]

\newtheorem{thm}[resultcounter]{Theorem}
\newtheorem{lem}[resultcounter]{Lemma}
\newtheorem{prop}[resultcounter]{Proposition}
\newtheorem{cor}[resultcounter]{Corollary}

\newtheorem{rem}[resultcounter]{Remark}
\newtheorem{rems}[resultcounter]{Remarks}
\newtheorem{exple}[resultcounter]{Example}

 \def\cB{{\cal B}} \def\cC{{\cal C}}
  
 \def\cH{{\cal H}} 
 \def\cK{{\cal K}}

 \def\cT{{\cal T}}

\newcommand{\R}{{\mathbb R}}
\newcommand{\N}{{\mathbb N}}

\newcommand{\C}{{\mathbb C}}
\newcommand{\Z}{{\mathbb Z}}

\renewcommand{\P}{{\mathbb P}}
\newcommand{\I}{{\mathbb I}}
\newcommand{\T}{{\mathbb T}}
\newcommand{\Ss}{{\mathbb S}}

\begin{document}
\title{Spectral Properties of Non-Unitary Band Matrices}

 \author{ Eman Hamza\footnote{Department of Physics, Faculty of Science, Cairo University, Cairo 12613, Egypt} \and Alain Joye\footnote{ UJF-Grenoble 1, CNRS Institut Fourier UMR 5582, Grenoble, 38402, France} }

%

\date{ }

\maketitle
\vspace{-1cm}

\thispagestyle{empty}
\setcounter{page}{1}
\setcounter{section}{1}

\setcounter{section}{0}

\abstract{We consider families of random non-unitary contraction operators defined as deformations of CMV matrices which appear naturally in the study of random quantum walks on trees or lattices. We establish several deterministic and almost sure results about the location and nature of the spectrum of such non-normal operators as a function of their parameters. We relate these results to the analysis of certain random quantum walks, the dynamics of which can be studied by means of iterates of such random non-unitary contraction operators.\\

\vspace{0.5cm}
\\\vspace{0.2cm} \hspace{0.4cm}{\bf Keywords.} Band matrices, Quantum random Walks.
}

\thispagestyle{empty}
\setcounter{page}{1}
\setcounter{section}{1}

\setcounter{section}{0}

\section{Introduction}

The spectral theory of self-adjoint and unitary operators is a well established topic in mathematics with a rich structure revealed by numerous important results, and which has found many applications, particularly in mathematical physics. See for example the textbooks \cite{Ka, RS, DS, D4, Ku} selected from the abundant literature on the topic. By contrast, the general spectral theory of operators enjoying less symmetry, that is non-normal operators, is more vast, technically more involved and less well understood. However, the spectral theory of non self-adjoint operators has been the object of many works, in various setups of regimes, as can be seen from the works \cite{GoKr,  SNF, D1, D2, TE, D3, Sj, CL, CCL, CD} and references therein.
In particular, several analyses of non self-adjoint operators focus on tri-diagonal operators, when expressed in a certain basis, see \cite{D1, D2, CL, CD}. Since Jacobi matrices provide generic models of self-adjoint operators, it is quite natural to deal with non self-adjoint tri-diagonal matrices which are deformations of Jacobi matrices. Moreover, certain models of this sort are physically relevant, see e.g. \cite{HN, GK, FZ}.

In this paper, we introduce and analyze the spectral properties of another set of  non-normal operators possessing a band structure in a certain basis, which share similarities with the tri-diagonal non-self-adjoint operators mentioned above. Our operators have a five-diagonal structure and are obtained as deformations of certain unitary operators called CMV matrices, see \cite{Si} for a detailed account. The role played by CMV matrices for unitary operators is similar to that played by Jacobi matrices for self-adjoint operators: they provide generic models of unitary operators; hence we call our models non-unitary operators. 
The non-unitary operators considered in this paper arise naturally in the study of random quantum walks on certain infinite graphs, which provide unitary dynamical systems of interest for physics, computer science and probability theory, see for example the reviews \cite{Ke, Ko, V-A, J4}. In particular, random quantum walks defined on $\Z$ are given by special cases of CMV matrices. The study of the spectral properties of random unitary operators and quantum walks defined on trees or lattices, see e.g.  \cite{bhj, HJS, JM, ASWe, J3, HJ}, may lead to the analysis of certain autocorrelation functions. We show  in Section \ref{rqw} below that in certain cases, the analysis of these autocorrelation functions reduces to the study of iterates of our non-unitary operators, which provides a direct link between spectral properties of non-unitary operators and random quantum walks. Moreover, the structure of our non-unitary operators allows us to determine the spectral nature of the corresponding random quantum walks they are related to.

 While the non-unitary operators we study correspond to deformations of random CMV matrices of a special type, and consequently are rather sparse, we show in Section \ref{sextension} that due to certain symmetries they possess, our main results also apply to deformations of  random unitary CMV type matrices of a much more general form. Those random unitary operators appear  as models in condensed matter physics and can be considered as natural unitary analogs of Anderson type models, see \cite{BB, bhj,  HJS}. The corresponding non-unitary deformations they give rise to are thus of a quite general form, displaying generically non zero elements at all entries of the familiar $5$-diagonal structure CMV type matrices possess. In that sense, our  spectral analysis applies to non-unitary deformations of typical random CMV type matrices addressed in the literature, which corresponds in this richer framework to the analyses of the non self-adjoint Anderson or Feinberg-Zee models addressed e.g. in \cite{D1, D2, CD}. 
\subsection{Main results}

The non-unitary operators $T_\omega$ addressed here are random operators on the Hilbert space $l^2(\Z)$ with the following structure: In the canonical basis of $l^2(\Z)$, denoted by $\{e_j\}_{j\in \Z}$, $T_\omega$ is defined as the infinite matrix 
\be\label{matrixt1}
T_\omega=\begin{pmatrix}
\ddots & e^{i\omega_{2j-1}}\gamma &e^{i\omega_{2j-1}} \delta & & & \cr
            &0                                          &0                                                   & & &\cr
             &0                                          &0      & e^{i\omega_{2j+1}}\gamma &e^{i\omega_{2j+1}}\delta &  \cr               
             & e^{i\omega_{2j+2}}\alpha &e^{i\omega_{2j+2}}\beta & 0& 0 &     \cr
             & & & 0                                          &0  &  \cr
           & &  & e^{i\omega_{2j+4}}\alpha &e^{i\omega_{2j+4}}\beta & \ddots
\end{pmatrix},
\ee
where the dots mark the main diagonal and the first column is the image of the vector $e_{2j}$. The phases $\{e^{i\omega_j}\}_{j\in \Z}$ are iid random variables and the deterministic coefficients, when arranged in a matrix $C_0\in M_2(\C)$, are constrained by the requirement that $C_0$ be a projection on $\C^2$ of a unitary matrix on $\C^3$:
\be\label{consti}
C_0=\begin{pmatrix} \alpha  & \beta  \cr  \gamma & \delta 
\end{pmatrix} \ \mbox{s.t.}\ \tilde C= \begin{pmatrix}
\alpha & r & \beta  \cr
q & g & s  \cr
\gamma & t & \delta  
\end{pmatrix} \in U(3), \ \mbox{with $0\leq g \leq 1$.}
\ee 
When $C_0$ itself is unitary, which corresponds to $g=1$, $T_\omega$  is a unitary random CMV matrix describing a random quantum walk, the spectral properties of which are known,  see \cite{JM, ASWe}. In general, however, $C_0$ is a contraction, and $T_\omega$ is a non-normal contraction, i.e. a non-unitary operator. We note here that, in general, $T_\omega$ is not a seminormal operator,  i.e. $[T_\omega^*,T_\omega]$ is not definite, see \cite{C}.
Non-unitary operators $T_\omega$ constrained by condition (\ref{consti}) appear as a natural objects in the study of the spectral properties of random quantum walks defined on the lattice $\Z^2$ or on $\cT_4$, the homogeneous tree of coordination number $4$, as explained in Section \ref{qwnuo}. This provides us with an independent motivation to focus on the characterization  (\ref{consti}) here, although other choices of deformations of CMV matrices are obviously possible. 
 Actually, Section \ref{sextension} shows that our spectral results extend to operators of the form $\widetilde T_\omega$ defined in the same basis as that used for (\ref{matrixt1}) by the random infinite matrix
\bea\label{matrixtt}
&&\widetilde T_\omega= \\ \nonumber
&& \begin{pmatrix}
\ddots & e^{i(\omega_{4j-1}+\omega_{4j-3})}\gamma\delta &e^{i(\omega_{4j-1}+\omega_{4j-3})}\delta^2 & & & \cr
            & e^{i(\omega_{4j-1}+\omega_{4j})}\gamma\beta  &e^{i(\omega_{4j-1}+\omega_{4j})}\delta\beta & & &\cr
             &e^{i(\omega_{4j+1}+\omega_{4j+2})}\gamma\alpha &e^{i(\omega_{4j+1}+\omega_{4j+2})}\gamma\beta       & e^{i(\omega_{4j+3}+\omega_{4j+1})}\gamma\delta &e^{i(\omega_{4j+3}+\omega_{4j+1})}\delta^2  &  \cr               
             & e^{i(\omega_{4j+2}+\omega_{4j+4})}\alpha^2 &e^{i(\omega_{4j+2}+\omega_{4j+4})}\alpha\beta & e^{i(\omega_{4j+3}+\omega_{4j+4})}\gamma\beta  &e^{i(\omega_{4j+3}+\omega_{4j+4})}\delta\beta &     \cr
             & & & e^{i(\omega_{4j+5}+\omega_{4j+6})}\gamma\alpha &e^{i(\omega_{4j+1}+\omega_{4j+2})}\gamma\beta  &  \cr
           & &  & e^{i(\omega_{4j+6}+\omega_{4j+8})}\alpha^2 &e^{i(\omega_{4j+6}+\omega_{4j+8})}\alpha\beta & \ddots
\end{pmatrix},
\eea
with entries characterised by (\ref{consti}). When $g=1$, the CMV type random operator $\widetilde T_\omega$ is unitary. The extension of our spectral analysis to the non-unitary deformation $\widetilde T_\omega$ is provided by  the identity  $\sigma (\widetilde T_\omega)=\sigma(T^2_\omega)$ and the spectral mapping theorem.

Our main spectral results about $T_\omega$ read as follows. After dealing with some special cases and with the translation invariant situation where $e^{i\omega_j}=1$, $j\in \Z$, we show in Theorem \ref{t1} that the polar decomposition of $T_\omega=V_\omega K$ has the following structure: the isometric part $V_\omega$ is actually unitary and has the same matrix structure as $T_\omega$, i.e. $V_\omega$ a one dimensional random quantum walk. Moreover, the self-adjoint part $K$ is deterministic with spectrum consisting in two infinitely degenerate eigenvalues $\{g,1\}$ only.   One consequence of this fact is that $T_\omega$ is a completely non-unitary contraction operator for $g<1$, so that the random quantum walk operator it comes from has no singular spectrum, see Proposition \ref{Hsing}.  This special structure also allows us to get informations on the spectrum of $T_\omega$ in terms of properties on $\sigma(V_\omega)$ and $\sigma(K)$, by applying a general result stated as Theorem \ref{gensigvk} and Corollary \ref{simplif}. This result  
determines parts of the resolvent set of a bounded operator of the form $T=AB$ with $A$, $B$ bounded, invertible and normal, in terms of the spectra of $A$ and $B$. A direct consequence is that the disc of radius $g>0$ centered at 0  is always contained in the resolvent set of $T_\omega=V_\omega K$ and, when $V_\omega$ contains a gap in its spectrum, other non-trivial explicitly determined sets also belong to $\rho(T_\omega)$, see Lemmas \ref{form} and \ref{form2}.

Then, we take advantage of the fact that the two spectral projectors of $K$ induce a natural bloc structure for $T_\omega$ which suggests the use of the Schur-Feshbach map. It turns out the blocs of the decomposition of $V_\omega$ are tridiagonal operators. This fact allows us to provide conditions on the parameter $g\in ]0, 1[$ in Theorem \ref{suff} which ensure that the spectrum of $T_\omega$ is contained in a centered ring with inner radius $g$ and outer radius strictly smaller than one. It also allows us to show in Lemma \ref{lcnu} that the circles of radii $1$ and $g$ cannot support any eigenvalues of $T_\omega$. These results are deterministic, but we further show that they hold for any realization of the random phases $\{e^{i\omega_j}\}_{j\in \Z}$. Finally, we take a closer look at the case $g=0$, the farthest to the unitary case, in some sense. Assuming the random phases are uniformly distributed and making use of ergodicity, we show that the almost sure spectrum of  $T_\omega$ consists in the origin and a centered ring whose inner and outer radii we determine. Also, in case the peripheral spectrum of $T_\omega$ coincides with the unit circle, we get that it contains no eigenvalue, whereas  the spectrum of  $V_\omega$ is pure point, and that of the corresponding random quantum walk operator is absolutely continuous, see Proposition \ref{g=0}.

The rest of the paper is organized as follows. Section \ref{rqw} provides a short summary of the relevant informations needed to make connection between the non-unitary operators $T_\omega$ considered in this paper and random quantum walks on $\cT_4$ and $\Z^2$. The link is made explicit in Section \ref{qwnuo}. The spectral properties of non-unitary operators is developed in the following two sections, together with the consequences which can be drawn for the random quantum walks they are related to and the explicit link  between $T_\omega$ and $\widetilde T_\omega$. The last  section is devoted to the case $g=0$. 

{\bf Acknowledgments } This work was supported in part by the French Government through a fellowship granted by the French Embassy in Egypt ( Institut Francais d'Egypte). E. H. thanks Universit\'e Grenoble-1 and the Institut Fourier where  this project was started, for support and hospitality. A.J. would like to thank J. Asch, Th. Gallay and S. Nonnenmacher for useful discussions. 

\section{Random Quantum Walks on $\Z^2$ and $\cT_4$}\label{rqw}

We provide here the basics on simple random quantum walks defined on the lattice $\Z^2$ and the homogeneous tree $\cT_4$, of coordination number $4$. Such quantum walks naturally depend on a $U(4)$-matrix valued parameter $C$ which drives the walk and monitors the effects of the disorder at the same time. In the next section, we focus on certain families of matrix valued parameters of interest which directly lead to the non-unitary operators $T_\omega$ considered in this paper. We also explain the consequences of our analysis of $T_\omega$ for the corresponding random quantum walks.

For more about  random quantum  walks and their spectral properties, we refer the reader to the reviews \cite{Ko, V-A, J4} and papers \cite{bhj, HJS, JM, J3,  HJ} and references therein.

We describe random quantum walks on the graph $\cT_4$ only according to \cite{HJ}, and will simply mention the occasional changes necessary to deal with the lattice case, as in \cite{J3}.

\subsection{Random quantum walks on $\cT_4$}
Let $\cT_4$ be a homogeneous tree of degree $4$, that we will consider as the tree of the free group generated by 
$
A_4=\{a,b, a^{-1}, b^{-1}\},
$ 
with $a a^{-1}=a^{-1}a=e=b b^{-1}=b^{-1}b$, $e$ being the identity element of the group; see Figure (\ref{tree4}). We choose a vertex of $\cT_4$ to be the root of the tree, denoted by $e$. Each vertex $x=x_{1}x_{2}\dots x_{n}$, $n\in\N$ of $\cT_4$ is a reduced word of finitely many letters from the alphabet $A_4$ and an edge of $\cT_4$ is a pair of vertices $(x,y)$ such that $xy^{-1}\in A_4$. 
The number of nearest neighbors of any vertex is thus $4$ and any pair of vertices $x$ and $y$ can be joined by a unique set of edges, or path in $\cT_4$. 
\begin{figure}[htbp]
\begin{center}
      \includegraphics[scale=.4]{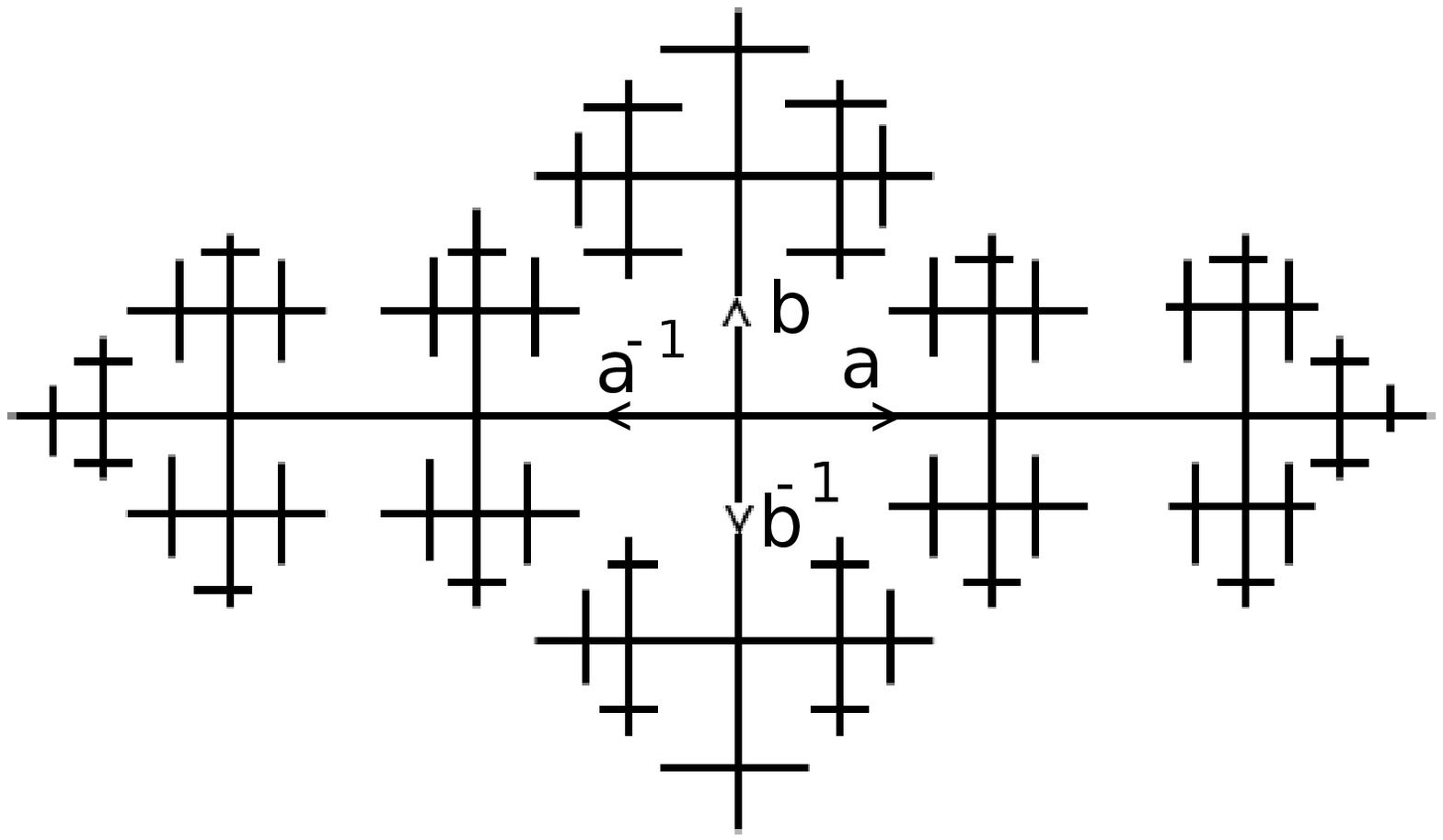}
   \end{center}\vspace{-.5cm}
   \caption{\footnotesize construction of $\cT_4$}\label{tree4}
\end{figure}
We identify $\cT_4$ with its set of vertices, and define the configuration Hilbert space of the walker by
$
l^2(\cT_4)=\Big\{\psi=\sum_{x\in \cT_4}\psi_x |x\ket \ \mbox{s.t.} \  \psi_x\in \C,  \ \sum_{x\in \cT_4}|\psi_x|^2<\infty\Big\},
$ 
where $|x\ket$ denotes the element of the canonical basis of $l^2(\cT_4)$ which sits at vertex $x$.
The coin Hilbert space (or spin Hilbert space) of the quantum walker on $\cT_4$ is $\C^4$. The elements of the ordered canonical basis of $\C^4$ are labelled by the letters of the alphabet $A_4$ as $\{|a\ket, |b\ket, |a^{-1}\ket, |b^{-1}\ket\}$.
The total Hilbert space is
\be\label{canbas}
\cK = l^2(\cT_4)\otimes \C^4 \ \mbox{ with canonical basis }\  \big\{x\otimes \tau\equiv |x\ket\otimes |\tau\ket ,\ \ x\in \cT_4, \tau\in A_4\big\}.
\ee
The quantum walk on the tree is characterized by the dynamics defined as the composition of a unitary update of the coin (or spin) variables in $\C^4$ followed by a coin (or spin) state dependent shift on the tree. 
Let $C\in U(4)$, $U(4)$ denoting the set of $4\times 4$ unitary matrices on $\C^4$. The unitary update operator given by $\I\otimes C$ acts on the canonical basis of $\cK$ as
\be\label{reshuffle}
(\I\otimes C) x\otimes \tau=|x\ket\otimes C|\tau\ket=\sum_{\tau'\in A_4} C_{\tau'\tau}\, x\otimes \tau',
\ee
where $\{C_{\tau'\tau}\}_{(\tau',\tau)\in A_4^2}$ denote the matrix elements of $C$.
The coin state dependent shift $S$ on $\cK$ is defined by
\bea\label{dirsumshift}
S={\sum}_{\tau\in A_4} S_\tau\otimes |\tau\ket\bra \tau| ,
\eea
where for all $\tau\in A_4$ the unitary operator $S_\tau$ is a shift that acts on $l^2(\cT_4)$ as  
$
S_\tau|x\ket=|x\tau\ket,  \forall x\in \cT_4
$,
with $S_\tau^{-1}=S_\tau^*=S_{\tau^{-1}}$. 
A quantum walk on $\cT_4$ is then defined as the one step unitary evolution operator on $\cK=l^2(\cT_4)\otimes \C^4$ given by 
\be\label{defeven}
U(C)=S(\I\otimes C)
={\sum_{\tau\in A_4\atop x\in \cT_4} |x\tau\ket\bra x |\otimes |\tau\ket\bra \tau| C},
\ee
where $C\in U(4)$ is a parameter.
A random quantum walk is defined via the following natural generalization. Let $\cC=\{C(x)\in U(4)\}_{x\in\cT_4}$ be a family of coin matrices  indexed by the vertices $x\in\cT_4$. A quantum walk with site dependent coin matrices is  defined by 
\bea\label{explicitsitedep}
U(\cC)&=&\sum_{\tau\in A_4, 
 x\in \cT_4} |x\tau\ket\bra x |\otimes |\tau\ket\bra \tau|  C(x).  
\eea

Consider  $\Omega=\T^{\cT_4\times A_4}$, $\T=\R/2\pi \Z$ the torus, as a probability space with $\sigma$ algebra generated by the cylinder sets and measure $\P=\otimes_{x\in\cT_4\atop \tau\in A_4}d\nu$ where $d\nu(\theta)=l(\theta)d\theta$, $l\in L^\infty(\T)$, is a probability measure on $\T$. 
Let $\{\omega^\tau_x\}_{x\in \cT_4, \tau\in A_4 }$ be a set of i.i.d. random variables on the torus $\T$ with common distribution $d\nu$. We will note $\Omega\ni \omega=\{\omega^\tau_x\}_{x\in \cT_4, \tau\in A_4 }$.
Our random quantum walks  are constructed by means of the following  families of site dependent random coin matrices: Let  $\cC_\omega=\{C_\omega(x)\in U(4)\}_{x\in \cT_4}$ be the collection of random coin matrices depending on a fixed matrix $C\in U(4)$, where, for each $x\in \cT_4$, $C_\omega(x)$ is defined by its matrix elements 
$ 
C_\omega(x)_{\tau \tau'}=e^{i\omega^{\tau}_{x\tau}}C_{\tau\tau'}, \ \ \tau,\tau'\in A_4^2.
$ 
The site dependence appears in the random phases only of the matrices $C_\omega(x)$, which have a fixed skeleton $C\in U(4)$. We consider random quantum walks 
defined by the operator
\be
U_\omega(C):=U(\cC_\omega) \ \mbox{on} \ \cK=l^2(\cT_4)\otimes \C^4
\ee
depending on $C\in U(4)$.
Defining a random diagonal unitary operator on $\cK$ by 
\be\label{dd}
\D_\omega x\otimes \tau = e^{i\omega^\tau_x} x\otimes \tau, \ \ \forall (x,\tau)\in \cT_4\times A_4,
\ee
we get that  $U_\omega(C)$ is manifestly unitary thanks to the identity
\be\label{iddef}
U_\omega(C)=\D_\omega U(C) \ \ \mbox{on } \cK.
\ee

\subsection{Random quantum walks on $\Z^2$}

The definition of a random quantum walk 
of the same type on $\Z^2$ instead of $\cT_4$ is the same, {\em mutatis mutandis}: the sites $x\in\cT_4$ are replaced by $x\in\Z^2$ so that the configuration space $l^2(\cT_4)$ is replaced by $l^2(\Z^2)$ but the coin space remains $\C^4$ in the definition of $\cK$. Thus the update operator $\I\otimes C$ is the same on $l^2(\Z^2)\otimes \C^4$ and on $l^2(\cT_4)\otimes \C^4$. Only the definition of the shifts $S_\tau$ in $S=\sum_{\tau\in A_4}S_\tau \otimes |\tau\ket\bra\tau|$, see (\ref{dirsumshift}), needs to be slightly changed.
We associate the letters $\tau$ of the alphabet $A_4$ with the canonical basis vectors $\{e_1, e_2\}$ of $\R^2$ as follows
$
a \leftrightarrow e_1,\ a^{-1} \leftrightarrow -e_1, \ b \leftrightarrow e_2,\  b^{-1} \leftrightarrow -e_2
$ 
and define the action of $S_\tau$ on $l^2(\Z^2)$ accordingly: for any $x=(x_1,x_2)\in \Z^2$,
$
S_a|x\ket=|x+e_1\ket, \ S_{a^{-1}}|x\ket=|x-e_1\ket, \ S_{b}|x\ket=|x+e_2\ket,\ S_{b^{-1}}|x\ket=|x-e_2\ket.
$ 
The random quantum walk is then defined by $U_\omega(C)$, as in (\ref{iddef}). 
\begin{rem}
All the results concerning $U_\omega(C)$ proven below for random quantum walks defined on $\cT_4$ hold for walks defined on $\Z^2$ as well, with the adaptations given above.
\end{rem}

\subsection{Spectral Criteria}

The main issue about random quantum walks concerns the long time behavior of the discrete random unitary dynamical system on the Hilbert space $\cK$ they give rise to by iteration of $U_\omega(C)$. The resulting dynamics is related to the spectral properties of $U_\omega(C)$  studied in the papers \cite{HJS, JM, ASWe, J3, HJ} on $\Z^d$ and $\cT_d$, as a function of $d\in \N$ and of the unitary matrix valued parameter $C$. We recall here well known spectral criteria which make a direct link between random quantum walks $U_\omega(C)$ on $\cT_4$ and $\Z^2$ and $T_\omega$ defined in (\ref{matrixt1}).

For a unitary operator $U$ on a separable Hilbert space $\cH$, the spectral measure $d\mu_\phi$ on the torus $\T$ associated with a normalized vector $\phi\in \cH$ decomposes as
$
d\mu_\phi=d\mu^{p}_\phi+d\mu^{ac}_\phi+d\mu^{sc}_\phi
$
 into its pure point, absolutely continuous and singular continuous components. The corresponding orthogonal spectral subspaces are denoted by $\cH^{\#}(U)$, with $\#\in\{p, ac, sc\}$. 
Then, see {\em e.g.} \cite{RS}, Wiener or RAGE Theorem relates the autocorrelation function ${n\mapsto\bra \phi | U^n \phi \ket}$ to the spectral properties of $U$:
\be
\lim_{N\ra\infty}\frac{1}{N}\sum_{n=0}^{N}|\bra \phi | U^n \phi\ket|^2=\sum_{\theta\in\T}(\mu^{p}_\phi\{\theta\})^2,
\ee
whereas the absolutely continuous spectral subspace of $U$, $\cH^{ac}(U)$, is given by
\be\label{critac}
\cH^{ac}(U)
=\overline{\Big\{\phi \ | \ \sum_{n\in \N}|\bra \phi | U^n \phi\ket|^2<\infty \Big\}}.
\ee

For example, consider $U(C)$ on $\cK$ given by (\ref{defeven}).  For any $C\in U(4)$, $\bra x\otimes\tau | U(C)^{2n+1} x\otimes~\tau \ket=0$ for any  $n\in\Z$ and $x\otimes \tau \in \cK$,  because $U(C)$ is off-diagonal.
Moreover, if $C=\I$,  $S=U(\I)$ further satisfies $\bra x\otimes \tau | S^{2n} x\otimes\tau \ket=\delta_{0,n}$, for all $x\otimes \tau \in \cK$, so that $d\mu_{x\otimes\tau}=\frac{d\theta}{2\pi}$ and $\sigma(S)=\sigma_{ac}(S)=\Ss$, the whole unit circle. The same holds for $U(\I)=S$ defined on $\Z^2$.
\section{Quantum Walks and Non-Unitary Operators}\label{qwnuo}

We consider here random quantum walks on $\cT_4$ characterized by coin matrices $C$ with a diagonal element of modulus one. As explained below, the non-trivial part of the dynamics they give rise to induces a systematic drift in one space direction. In other words, the dynamics induces a leakage of the wave vectors in one direction that is associated with a purely absolutely continuous part of spectrum of the corresponding evolution operator.
We approach this spectral question by analysing the restriction of $U_\omega(C)$ to a one-dimensional subspace that defines the random contractions $T_\omega$ we study in this paper. The consequences for such quantum walks of our results about the contractions $T_\omega$, namely the proof that the evolution operator is purely absolutely continuous for all realisations of the disorder, are spelled out in Lemma \ref{lemspr} and Proposition \ref{Hsing}.
Finally, we note that from the perspective of the determination of the spectral phase diagram for random quantum walks  on $\cT_4$, the corresponding set of coins matrices is not covered by the work \cite{HJ}. 

Without loss, we assume that the coin matrix $C$ with a diagonal element of modulus one takes the following form in the ordered basis  $\{|a\ket, |b\ket,|a^{-1}\ket,|b^{-1}\ket\}$, 
\be\label{lcm}
C=\begin{pmatrix}
\alpha & r & \beta & 0 \cr
q & g & s & 0 \cr
\gamma & t & \delta & 0 \cr
0 & 0 & 0 & e^{i\theta} \cr
\end{pmatrix}\equiv
\begin{pmatrix} \tilde C & \vec 0 \cr \vec 0 ^T & e^{i\theta}\end{pmatrix}\in U(4),
\ \ \mbox{where}\  \
\tilde C= \begin{pmatrix}
\alpha & r & \beta  \cr
q & g & s  \cr
\gamma & t & \delta  
\end{pmatrix}\in U(3),
\ee 
with $\theta\in \T$ and $1\geq g\geq 0$.
The assumption $g\geq 0$ always holds at the price of a multiplication of $C$, and thus of $U_\omega(C)$, by a global phase which does not affect the spectral properties.
By construction, $U_\omega(C)$ admits $\cK^{b^{-1}}$, the subspace characterized by a coin variable equal to $|b^{-1}\ket$, as an invariant subspace on which it acts as the shift $S_{b^{-1}}$, up to phases. Hence
\be
\sigma\left(U_\omega(C)|_{\cK^{b^{-1}}}\right)=\sigma_{ac}\left(U_\omega(C)|_{\cK^{b^{-1}}}\right)=\Ss.
\ee
Let $\cK^{\perp}$ be the complementary invariant subspace 
\be
\cK^{\perp}=\overline{\mbox{span }}\Big\{x\otimes \tau \ | \ x\in \cT_4, \tau\in \{a,b,a^{-1}\}\Big\},
\ee
where the notation $\overline{\mbox{span }}$ means the closure of the span of vectors considered.
On $\cK^{\perp}$ the action of $U_\omega(C)$ on the quantum walker makes it move horizontally back and forth, but it only makes it go up vertically, see Figure (\ref{tree4}). In a sense, the dynamics induces a leakage of the vectors in the direction corresponding to the coin state $|b\ket$.
In order to assess that $U_\omega(C)|_{\cK^{\perp}}$ has purely absolutely continuous spectrum, an application of criterion (\ref{critac}) leads us to consider $\bra \psi | U_\omega(C)^n\psi\ket$, $n\geq 0$, with normalized vector $\psi\in\cK^{\perp}$. Note that by construction, for all $x\in \cT_4$, all $\tau\in \{a,b,a^{-1}\}$
\be\label{escape}
\bra x\otimes b | U_\omega(C)^n x\otimes \tau\ket = \delta_{n,0}\delta_{b,\tau}, \ \ \forall \ n\in \N, \forall \ x\in \cT_4.
\ee
In particular, all spectral measures $d\mu_{x\otimes b}(\theta)=\frac{d\theta}{2\pi}$ on $\T$ and $\sigma\left(U_\omega(C)|_{\cK^{\perp}}\right)=\Ss$ as well. We thus have,
\be\label{hbdef}
\cH_b=\overline{\mbox{span }}\Big\{x\otimes \tau \ | \ x\in \cT_4, \tau\in \{b,b^{-1}\}\Big\}\subset \cH^{ac}(U_\omega(C)).
\ee
\subsection{Reduction to One Space Dimension}

To this end we introduce the horizontal subspace associated with the direction $a$
\be\label{p0}
\cH_0=\overline{\mbox{span }}\Big\{x\otimes \tau \ | \ x=a^m\in \cT_4, m\in\Z, \tau\in \{a,a^{-1}\}\Big\}\subset \cK^\perp\subset \cK,
\ee
and 
$P_0: \cK\ra \cK$, the orthogonal projector onto $\cH_0$. All vectors in this subspace live on the horizontal one dimensional lattice passing through the root of $\cT_4$. We can actually consider vectors on any other horizontal one dimensional lattice by attaching $\cH_0$ to any other vertex. 
To study $P_0U_\omega(C)^nP_0$, $n\geq 0$ we first note the following simple lemma which allows us to focus on the restriction of $U_\omega(C)$ to $\cH_0$.

\begin{lem}\label{contraction} Let $T_\omega: \cH_0\ra \cH_0$ be defined by
$
T_\omega=P_0U_\omega(C)P_0|_{\cH_0}
$ and $T=T_\omega|_{\omega=(\cdots, 0, 0, 0,\cdots)}$. 
Then, $T_\omega$ is a contraction,
\be\label{deft}
T_\omega=\D^0_\omega T,\ \mbox{where } \ \D^0_\omega=\mbox{\em diag }(e^{i\omega_x^\tau}), 
\ee
is the restriction of (\ref{dd}) to $\cH_0$, 
and, for any $n\in \N$,
$
P_0U_\omega(C)^nP_0|_{\cH_0}=T_\omega^n.
$
\end{lem}
\begin{proof} First, we have $\|T_\omega\|=\|P_0U_\omega(C)P_0\|\leq 1$ and  $[\D_\omega,P_0]=0$ proves the second statement.
Set $Q_0=\un -P_0$ and let us show that for all $k\geq 1$, $P_0U_\omega(C)^kQ_0U_\omega(C)P_0=0$. Indeed, for any basis vector $x\otimes\tau$ of $\cH_0$, $Q_0U_\omega(C)x\otimes\tau$ is proportional to $xb\otimes b$, where $xb\neq a^m$, for all $m\in \Z$. Consequently, $P_0U_\omega(C)^kxb\otimes b=0$, for any $k\geq 1$, which yields the result. \ep
\end{proof}
\begin{rems}\label{tbloc} 
i) The contraction $T$ can be written according to  (\ref{defeven}) as 
\be\label{quasiwalk}
T=S(\I\otimes C_{0})=S_a\otimes |a\ket\bra a| C_{0}+S_{a^{-1}}\otimes |a^{-1}\ket\bra a^{-1}| C_{0},
\ee
where $C_0=\Pi_0 C \Pi_0|_{\Pi_0\C^4}$ with $\Pi_0=|a\ket\bra a|+ |a^{-1}\ket\bra a^{-1}|$ is a contraction which takes the form
\be\label{czero}
C_0=\begin{pmatrix} \alpha  & \beta  \cr  \gamma & \delta 
\end{pmatrix} \  \mbox{ in the ordered basis $\{|a\ket,|a^{-1}\ket\}$.}
\ee
We will say that $C_0$ characterizes the operator $T$.\\
ii) Such an operator, or its higher dimensional analogs, define {\em contractive quantum walks}.\\
\end{rems}

Since $T_\omega$ is not normal in general, the inequalities
$
\mbox{spr }(T_\omega)\leq \|T_\omega\|\leq 1
$
are not necessarily  saturated. Actually, we prove below, Corollary \ref{normt}, that $\|T_\omega\|=1$, so that we need to extract spectral informations about $T_\omega$ in order to get decay as $n\ra \infty$ of the autocorrelation function $|\bra \psi | U_\omega(C)^n \psi\ket|$, $\psi\in \cH_0$.  Hence,
\begin{lem}\label{lemspr} With the notations above,
$ 
 \mbox{\em spr }(T_\omega)<1 \ \Rightarrow \ U_\omega(C) \ \mbox{is purely ac}, \ \forall \omega\in \Omega.
$ 
\end{lem}
\begin{proof}
 If the spectral radius of $T_\omega$ satisfies $\mbox{spr }(T_\omega)<1$, then, for any 
$\eps>0$ s.t. 
$|\ln (\mbox{spr }(T_\omega))|-\eps>0$, $\|T^n_\omega\|\leq (\mbox{spr }(T_\omega)e^{\eps})^n$, if $n$ is large enough. Thus, for any normalized $\psi\in \cH_0$, we have
$
|\bra \psi | U_\omega^n \psi\ket| = |\bra  \psi | T_\omega^n \psi\ket|\leq e^{-n(|\ln (\mbox{\scriptsize spr }(T_\omega))|-\eps)},  \ \mbox{if $n$ is large enough}. 
$
Thus $\cH_0\subset \cH^{ac}(U_\omega)$. Since $\cH_0$ can be attached to any vertex of the tree, we get the result.\ep
\end{proof}
\begin{rem} We show below in Proposition \ref{Hsing} that a finer analysis of the structure of $T_\omega$ implies that $U_\omega(C)$ is purely ac for all $\omega$, if $g<1$.\end{rem}
\section{One-Dimensional Contractive Quantum Walk}

We turn to the analysis of the random contractive quantum walk defined by 
(\ref{deft}) and (\ref{quasiwalk}) with parameters 
\be\label{const}
C_0=\begin{pmatrix} \alpha  & \beta  \cr  \gamma & \delta 
\end{pmatrix} \ \mbox{s.t.}\ \tilde C= \begin{pmatrix}
\alpha & r & \beta  \cr
q & g & s  \cr
\gamma & t & \delta  
\end{pmatrix}\in U(3) \ \mbox{and}\ 0\leq g \leq 1. \ee 
We view this problem as a question of independent interest in the spectral analysis of non self-adjoint or, more adequately in the present context,  non-unitary operators.

We start by the following simple property relating $C_0$ to $\tilde C$.
\begin{lem}\label{gmg}
Let $C_0=\begin{pmatrix} \alpha  & \beta  \cr  \gamma & \delta 
\end{pmatrix}$ be a contraction on $\C^2$ which is not unitary. Then, there exists $\tilde C\in U(3)$ such that (\ref{const}) holds.
\end{lem}
\begin{proof} By exchanging the basis vectors, we can look for $\tilde C$ in the bloc form $\tilde C=\begin{pmatrix} C_0  & u  \cr \bar v^T  & g 
\end{pmatrix}$, where $u, v$ denote vectors  in $\C^2$ and $g\in [0,1]$. Imposing that $\tilde C\in U(3)$, we get, 
\bea
C_0^*C_0&=& \I_{\C^2} -  | v\ket\bra  v|, \ \ 
\|v\|^2=1-g^2, \ \ 
C_0v=-gu\\ \nonumber
C_0C_0^*&=& \I_{\C^2} -  | u\ket\bra  u|, \ \ 
\|u\|^2=1-g^2, \ \ 
C_0^*u=-gv.
\eea
It follows that $\sigma(C_0^*C_0)=\{1,g^2\}$, which determines $0\leq g<1$ and the norm of the corresponding eigenvector $v$ of $C_0^*C_0$. If $g\neq 0$, then $u=-C_0v/g$. In case $g=0$, $u$ is a normalized eigenvector of $\ker C_0^*$. \ep
\end{proof}

Identifying the subspace $\cH_0$ with $l^2(\Z)$, we get a representation of $T_\omega$ by a 5-diagonal doubly infinite matrix.
Let $\{e_j\}_{j\in\Z}$, resp.  $\{a^m\otimes \tau\}_{m\in \Z}^{\tau\in\{a, a^{-1}\}}$, be the canonical orthonormal basis of $l^2(\Z)$, resp. $\cH_0$. We map the latter to the former according to the rule
\be\label{rule}
e_{2j} = a^j\otimes a, \ \ e_{2j+1} = a^j\otimes a^{-1}, \ \ j\in \Z
\ee 
and relabel the random phases $\omega_x^\tau$ accordingly,
so that we can identify $T_\omega$ with the matrix
\be\label{matrixt}
T_\omega=\D_\omega^0T=\begin{pmatrix}
\ddots & e^{i\omega_{2j-1}}\gamma &e^{i\omega_{2j-1}} \delta & & & \cr
            &0                                          &0                                                   & & &\cr
             &0                                          &0      & e^{i\omega_{2j+1}}\gamma &e^{i\omega_{2j+1}}\delta &  \cr               
             & e^{i\omega_{2j+2}}\alpha &e^{i\omega_{2j+2}}\beta & 0& 0 &     \cr
             & & & 0                                          &0  &  \cr
           & &  & e^{i\omega_{2j+4}}\alpha &e^{i\omega_{2j+4}}\beta & \ddots
\end{pmatrix},
\ee
where the dots mark the main diagonal and the first column is the image of the vector $e_{2j}$.
We note three special cases which allow for a complete description of the spectrum of $T_\omega$.
\begin{lem}\label{sc}
If $\alpha=\delta=0$,  the subspaces $\mbox{\em span }\{e_{2j+1}, e_{2j+2}\}$ reduce $T_\omega$. We have
\be
T_\omega=\oplus_{j\in\Z}T_\omega^{(j)}, \ \mbox{where } \ T_\omega^{(j)}=\begin{pmatrix} 
0 & \gamma e^{i\omega_{2j+1}}\cr  \beta e^{i\omega_{2j+2}} & 0
\end{pmatrix},\ j\in \Z,
\ee
 $\sigma(T_\omega)=\cup_{j\in \Z}\{\pm g^{1/2} e^{i\theta/2}e^{i(\omega_{2j+1}+\omega_{2j+2})/2}\}$, and $g=\min{(|\beta|,|\gamma|)}$, $\theta=\arg (\beta\gamma)$.

If $\beta=\gamma=0$, the subspaces $\cH_+=\overline{\mbox{span }}\{e_{2j}\}_{ j\in \Z}$ and $\cH_-=\overline{\mbox{span }}\{e_{2j+1}\}_{ j\in \Z}$ reduce $T_\omega$. We have, with $S_\pm$ the standard shifts on $\cH_\pm$,
\be
T_\omega=T^{(+)}_\omega\oplus T^{(-)}_\omega, 
\ee
where, $\ T^{(+)}_\omega=T_\omega|_{\cH_+}$ is unitarily equivalent to $|\alpha|S_+$, similarly  $\ T^{(+)}_\omega=T_\omega|_{\cH_+}$ is unitarily equivalent to $|\delta|S_-$.
 $\sigma(T_\omega)=\Ss\cup g\Ss$, and $g=\min(|\alpha|, |\delta|)$.

If $g=1$, $T_\omega$ is unitary with $\sigma_c(T_\omega)=\emptyset$, almost surely, 
unless $C_0\in U(2)$ is diagonal, in which case $\sigma(T_\omega)=\sigma_{ac}(T_\omega)=\Ss.$
\end{lem}
\begin{proof}
The decompositions of $T_\omega$ under the assumptions made is straightforward. The only point is the determination of the spectral radius when the coefficients are constrained by (\ref{const}). We consider $\alpha=\delta=0$ only, the other case being similar. In such a case (\ref{const}) implies $\bar qs=0$ so that either $q=t=0$ or $s=r=0$. In which case $|\gamma|=1$, or $|\beta|=1$. In the first case, $g^2+|r|^2=1=|r|^2+|\beta|^2$, so that $g=|\beta|=\min(|\beta|, |\gamma|)$. The case $|\beta|=1$ is similar. Finally, the case $g=1$ implies that $C_0$ is unitary, so that $T_\omega$ is a one dimensional random quantum walk, and \cite{JM} applies to yield the result. \ep
\end{proof}
\begin{rem}\label{highdim}
Quantum walks of the general form (\ref{iddef}) can be defined on $\Z^d$ or $\cT_{2d}$, with $d\in \N$, using the obvious extension to higher dimensions, see \cite{HJ}.  When reduced to a one dimensional lattice of the form $\cH_0$, they give rise to a contractive quantum walk which has the form of a CMV type matrix of the kind (\ref{matrixt}). In general, $U(\cC)$ is not a dilation of the corresponding contractive quantum walk. However, if the quantum walk $U_\omega(C)$ defined on $\cT_{2d}$, say, with coin matrix $C\in U(2d)$ 
having similar properties as for $d=2$, this property is still true:  let us denote the coin states basis by $\{ |a_j\ket, |a_j^{-1}\ket\}_{j=1,\dots, d}$ and assume $C |a_j^{-1}\ket=e^{-i\theta_j}|a_j^{-1}\ket$, for $j=2,\dots,d$. Consider the subspace $\cH_0$ associated with the direction $a_1$ and $P_0$ the corresponding orthogonal projection onto $\cH_0$; then $U_\omega(C)$ is a dilation of the contraction $T_\omega= P_0 U_\omega(C) P_0$,  i.e. Lemma \ref{contraction} holds. 
\end{rem}

\subsection{Translation invariant case}
The deterministic, translation invariant case characterized by $\D_\omega=\un$, i.e. $T_\omega=T$, is best tackled by Fourier methods. We map  $l^2(\Z)$ unitarily onto $ L^2(\T;\C^2)$ via the identification 
\be
\psi=\sum_{j\in\Z}c_j|j\ket\in l^2(\Z) \ \leftrightarrow \ \ f(x)=\begin{pmatrix}f_+(x)\cr f_-(x)\end{pmatrix}\in L^2(\T;\C^2),
\ee where $f_+(x)=\sum_{j}c_{2j}e^{i2jx}$, $f_-(x)=\sum_{j}c_{2j+1}e^{i(2j+1)x}$, $x\in \T$. Then $T$ is unitarily equivalent on $L^2(\T;\C^2)$ to the multiplication operator by the analytic matrix valued function
\be\label{matt}
T\simeq T(x)=\begin{pmatrix} \alpha e^{i2x} & \beta e^{ix} \cr \gamma e^{-ix} & \delta e^{-i2x}
\end{pmatrix}.
\ee 
The following criteria for more symmetries hold true.
\begin{lem}\label{symlem}
i) $T$ is self-adjoint $\Leftrightarrow$   $C_0=\begin{pmatrix} 0  & e^{i\nu}  \cr   e^{-i\nu} & 0 
\end{pmatrix}$, $\nu\in \R$. This implies $g=1$, $T$ is unitary and $\sigma(T)=\{-1,1\}$.\\
ii) $T_\omega$ is unitary $\Leftrightarrow$ $|\det C_0|=\left|\det \begin{pmatrix} \alpha  & \beta  \cr \gamma  & \delta 
\end{pmatrix}\right|=1$.
\end{lem}
\begin{proof} 
We have $T$ is sef-adjoint if and only if $T(x)$ is self-adjoint for all $x\in \T$, which 
together with (\ref{const}) readily implies the first statement. The second statement 
 is a consequence of the general simple lemma
\begin{lem}\label{lemgen}
Let $W\in M_d(\C)$ be a contraction. Then, $W$ is unitary $\Leftrightarrow$ $|\det(W)|=1$.
\end{lem}
Indeed, $T_\omega$ is unitary if and only if $T$ is unitary, which is true, see (\ref{quasiwalk}) if and only if $C_0$ is unitary, and the lemma applies to the last matrix valued contraction. \\
\begin{proof}
The direct implication is trivial.
Assume $|\det(W)|=1$ and consider the spectral decomposition 
\be
W=\sum_{k=1}^m\lambda_kP_k+D_k, 
\ee
where $\sigma(W)=\{\lambda_k\}_{1\leq k\leq m}$, and $\{P_k\}_{1\leq k\leq m}$, resp. $\{D_k\}_{1\leq k\leq m}$, are the eigenprojectors, resp. eigennilpotents of $W$. Since $W$ is a contraction the condition on the  determinant implies  $|\lambda_k|=1$, $k=1, 2, \dots, m$. Moreover, $\|W^n\|\leq 1$ for all $n\geq 0$, so that all eigennilpotents are equal to zero, since 
\be
W^n=\sum_{k=1}^m\lambda_k^nP_k+\sum_{r=0}^KD_k^r\lambda_k^{n-r}\begin{pmatrix} n \cr  r
\end{pmatrix}, \ \ \mbox{$n\geq K$},
\ee
where $K$ is the maximal index of nilpotency of the $D_k's$. Eventually, the general property $\|P_k\|\geq 1$ together with $\sigma(W)\subset \Ss$  imply that $\|P_k\|=1$ for $W$ to be a contraction, so that $P_k=P_k^*$ for all $k=1, 2, \dots, m$. \ep
\end{proof}
\end{proof}

As $T$ is unitarily equivalent to a multiplication operator, its spectrum is readily obtained in the generic case. For all $x\in \T$, consider the eigenvalues of $T(x)$
\be\label{evti}
\lambda_{\pm}(x)=\frac12\left(\alpha e^{i2x}+\delta e^{-i2x}\pm\{(\alpha e^{i2x}+\delta e^{-i2x})^2-4(\alpha\delta-\beta\gamma)\}^{1/2}\right).
\ee
Assume that $\T\cap Z=\emptyset$, where  $Z=\{x\in \C \ | \ \lambda_-(x)=\lambda_+(x)\}$ is the finite set of exceptional points $T(x)$, see \cite{Ka}. Then, with $P_\pm(x)$ the eigenprojectors of the diagonalizable matrix $T(x)$, we get that  $(T-z)^{-1}$ is given for $z\in\rho(T)$ by the multiplication operator 
$
R_z(x)=\frac{P_-(x)}{\lambda_-(x)-z}+\frac{P_+(x)}{\lambda_+(x)-z},
$ 
on $L^2(\T; \C^2)$ 
and $\sigma (T)=\mbox{Ran }\lambda_-\cup \mbox{Ran }\lambda_+$. 
\subsection{Polar decomposition of $T_\omega$}

In case the contractive quantum walk $T_\omega$ is random, we cannot use Fourier transform methods to determine  $\mbox{spr} (T_\omega)$ but, instead, we resort to the properties of its polar decomposition.
Let us come back to the general case (\ref{matrixt}) and consider the unique decomposition  $T_\omega=V_\omega K_\omega$, where $K_\omega$ is a non negative operator  on $l^2(\Z)$ and $V_\omega$ is an isometry on $l^2(\Z)$.  We note that due to (\ref{deft}), $K_\omega$ is independent of the randomness since $T_\omega^*T_\omega = T^*T=K^2$. 

\begin{thm}\label{t1}
The contraction $T_\omega$ defined on $l^2(\Z)$  by (\ref{matrixt}) with the constraint (\ref{const}) admits the polar decomposition $T_\omega=V_\omega K$, where $0\leq K\leq \un$ is given by
\be\label{specdeck}
K=P_1+gP_2,  \ \ \mbox{with} \ \ \sigma(K)=\sigma_{ess}(K)=\{1,g\} \ \ \mbox{and} \ \ \|K\|=1,
\ee 
and with infinite dimensional spectral projectors $P_j$, $j=1,2$ given in (\ref{projk}) below.\\
The isometry $V_\omega$ is unitary on $l^2(\Z)$ and takes the form $V_\omega=\D_\omega^0 V$, with 
\be\label{matv}
V=\frac{1}{1+g}\begin{pmatrix}
\ddots & \gamma(1+g)-qt &\delta(1+g)-st & & & \cr
            &0                                          &0                                                   & & &\cr
             &0                                          &0      & \gamma(1+g)-qt  &\delta(1+g)-st  &  \cr               
             & \alpha(1+g)-qr &\beta(1+g)-sr & 0& 0 &     \cr
             & & & 0                                          &0  &  \cr
           & &  & \alpha(1+g)-qr &\beta(1+g)-sr & \ddots
\end{pmatrix},
\ee
where the dots mark the main diagonal and the first column is the image of the vector $e_{2j}$.
\end{thm}
\begin{cor}\label{normt}  for all $\omega\in \Omega$,  $T_\omega$ satisfies:
$
\|T_\omega\|=1 \  \mbox{and } \ T_\omega\,  \mbox{is unitary}  \Leftrightarrow  g=1.
$
\end{cor}
\begin{rems}\label{detg} i) Condition (\ref{const}) implies $g=\left|\det \begin{pmatrix} \alpha  & \beta  \cr \gamma  & \delta 
\end{pmatrix}\right|$.\\
ii) The unitary operator $V$ corresponds to a one-dimensional quantum walk with unitary coin matrix  $\frac{1}{1+g}\begin{pmatrix} \alpha(1+g)-qr & \beta(1+g)-sr\cr \gamma(1+g)-qt & \delta(1+g)-st \end{pmatrix}$, according to Remark \ref{tbloc}. \\ iii) The random quantum walk $V_\omega$ displays dynamical localization for all values of the parameters in (\ref{const}), unless the coin matrix is diagonal, in which case it is absolutely continuous, see \cite{JM}.\\
iv) When $g=1$, the original random quantum walk characterized by (\ref{lcm})  decouples into one-dimensional problems the solutions of which are known, \cite{JM}. Thus, we assume $0\leq g<1$.\\
v) We have $0\in \sigma(K)$ iff \ $0\in \sigma(T)$, and $\ker \, K=\ker \, T$, since $V$ is unitary.
\end{rems}

The proof of Theorem \ref{t1} entails explicit computations of $K$ and $V_\omega$ which are detailed in the next two propostions.

\begin{prop}\label{k2} Assume $0\leq g<1$. The two-dimensional orthogonal subspaces $\cH^{(k)}=\mbox{span}\{e_{2k}, e_{2k+1}\}$ reduce the operator $K=(T^*T)^{1/2}$ which takes the form 
\be\label{redk}
K=\bigoplus_{k\in \Z}\kappa_k \ \ \mbox{with respect to  } \ \ \cH_0=\bigoplus_{k\in\Z}\cH^{(k)}.
\ee
The bloc $ \kappa_k$ acts in the ordered basis $\{e_{2k}, e_{2k+1}\}$ as
\be
\kappa_k= \frac{1}{|q|^2+|s|^2}\begin{pmatrix} g|q|^2+|s|^2 & \bar q s(g-1) \cr q \bar s (g-1) & g|s|^2+|q|^2 \end{pmatrix}, \ \ \forall k\in \Z,
\ee
see   (\ref{const}).
The spectral decomposition of $\kappa_k$ reads
\be
\kappa_k=Q^{(k)}_1+g Q^{(k)}_2, \ \ \mbox{where } \ \ Q^{(k)}_1=\frac{1}{|q|^2+|s|^2}\begin{pmatrix} |s|^2 & -\bar q s \cr -q \bar s & |q|^2 \end{pmatrix}=\un_{2}-Q^{(k)}_2.
\ee
\end{prop}
We deduce the spectral decomposition of $K$ given in Theorem \ref{t1} immediately:
\be\label{projk}
\sigma(K)=\{1,g\}, \ \ K=P_1+gP_2, \ \ \mbox{where } \ \ P_j=\bigoplus_{k\in \Z} Q_j^{(k)}, \ j=1,2.
\ee
\begin{proof}
A straightforward computation based on definition (\ref{matrixt}) yields
\be
K^2=\bigoplus_{k\in \Z}\begin{pmatrix} |\alpha|^2+|\gamma|^2 & \delta \bar \gamma +\beta \bar \alpha  \cr \gamma  \bar \delta + \alpha \bar \beta  & |\beta|^2+|\delta|^2 \end{pmatrix}\equiv \bigoplus_{k\in \Z} \kappa_k^2
\ee
with the decomposition of $\cH_0$ given by (\ref{redk}). Condition (\ref{const}) allows us to rewrite the blocs $\kappa_k^2$ of this decomposition as
\be
\kappa_k^2=\begin{pmatrix} 1-|q|^2 & -  s \bar q  \cr -  q \bar s & 1-|s|^2 \end{pmatrix}, \ \ \mbox{ with } \left\{\begin{matrix} \hspace{-.5cm}\det \kappa_k^2=1-(|q|^2+|s|^2)=g^2 \cr \tr \kappa_k^2=2-(|q|^2+|s|^2)=1+g^2.\end{matrix}\right.
\ee
Hence, $\sigma(\kappa_k^2)=\{1, g\}$ with corresponding normalized eigenvectors 
\bea\label{evek}
v^{(k)}_1=\frac{1}{\sqrt{|q|^2+|s|^2}}\begin{pmatrix} s \cr -q \end{pmatrix},  \ \ 
v^{(k)}_2=\frac{1}{\sqrt{|q|^2+|s|^2}}\begin{pmatrix} \bar q \cr \bar s \end{pmatrix}.
\eea
Explicit computations yield the spectral projectors $Q^{(k)}_1=|v^{(k)}_1\ket\bra v^{(k)}_1|$ and $Q^{(k)}_2=\un_2-Q^{(k)}_2$, and, in turn, $\kappa_k={(\kappa_k^2)}^{1/2}$. The spectral decomposition of $K$ follows immediately.\ep
\end{proof}

We now turn to the computation of the isometry $V_\omega=\D_\omega^0 V$. Recall that  translation invariant operators with the same band structure matrix as $T$ are characterized by a $2\times 2$ matrix, in the same way as $T$ is characterized by $\begin{pmatrix} \alpha  & \beta  \cr  \gamma & \delta 
\end{pmatrix}$, see Remark \ref{tbloc}.
\begin{prop}
For $1>g>0$, $V=TK^{-1}$ where $K^{-1}=\bigoplus_{k\in \Z}\kappa_k^{-1}$ and 
\be
\kappa_k^{-1}=\frac{1}{g(1+g)}\begin{pmatrix}
1-|s|^2+g & s\bar q\cr q \bar s &  1-|q|^2+g
\end{pmatrix}.
\ee
The operator $V$ has the same band structure as $T$ and is  characterized by the unitary matrix
\be\label{simv}
\begin{pmatrix} \alpha  & \beta \cr \gamma & \delta  
\end{pmatrix}\kappa_k^{-1}=\frac{1}{1+g}\begin{pmatrix}   \alpha(1+g)-qr & \beta(1+g)-sr\cr \gamma(1+g)-qt & \delta(1+g)-st\end{pmatrix}.
\ee
\end{prop}
\begin{rem}
The unitary operator $V$ is well defined in the limit $g\ra 0$, with the constraint (\ref{const}), even though $K^{-1}$ is not.
\end{rem}
\begin{proof} The first statement is a consequence of Proposition \ref{k2} and of the spectral theorem. The invariance of the subspaces $\mbox{span}\{e_{2k}, e_{2k+1}\}$ under $K^{-1}$ and the matrix structure of $T$ imply that $V$ has the same structure as $T$. It is a matter of computation to check statement (\ref{simv}), systematically using constraint (\ref{const}) to simplify the factor $g$ in the denominator.\ep
\end{proof}

\subsection{ Structure of the Contraction $T_\omega$}\label{scnu}
 Recall that a contraction is said to be completely non-unitary, cnu for short,  if it possesses no non-trivial closed invariant subspace on which it is unitary, see {\em e.g.} \cite{SNF}.
\begin{lem}\label{lcnu} Let $0\leq g <1$. Then, for all  $\omega\in \Omega$, the operator $T_\omega$ is either cnu or it is  unitarily equivalent to the direct sum of a shift and of $g$ times a shift.
Consequently,
\be
\sigma_p(T_\omega)\cap \Ss =\emptyset, \ \ \mbox{and for $0<g<1$,}\ \ \sigma_p(T_\omega)\cap g\Ss=\emptyset.
\ee
\end{lem}
\begin{proof}
Assume there is a closed subspace $\h_0$  such that $T_\omega|_{\h_0}$  is unitary.
For $\psi\in \h_0$, we have $\|T_\omega\psi\|=\|\psi\|$. This implies with  $T_\omega=V_\omega(P_1+gP_2)$, that 
\be\label{cnu}
(\I-T_\omega^*T_\omega)^{1/2}\psi=\sqrt{1-g^2}P_2\psi=0.
\ee 
Hence, $\h_0\subset P_1 \cH_0$, and, $\h_0$ being invariant under $T_\omega$, $\h_o\subset \ker P_2 V_\omega P_1.$
The operator $ P_2 V_\omega P_1$ is studied in Lemmas \ref{qlem} and \ref{loffdiag} below, where it is shown that  $ \ker P_2V_\omega P_1\neq \{0\} \Leftrightarrow P_2V_\omega P_1=0$ and that this is 
equivalent to
\bea\label{speci} \tilde{C}\in \left\{ 
\begin{pmatrix}
\alpha & r & 0  \cr
q & g & 0  \cr
0 & 0 & \delta  
\end{pmatrix}, 
\begin{pmatrix}
\alpha & 0 & 0  \cr
0 & g & s  \cr
0 & t & \delta  
\end{pmatrix}
\right\}\subset U(3).
\eea Hence if (\ref{speci}) doesn't hold, $T_\omega$ is cnu, whereas in case
(\ref{speci}) holds, Lemma \ref{sc} finishes the proof of the first statement. The fact that eigenvalues cannot sit on the unit circle is thus immediate, whereas, for $g>0$, a similar argument applied to the contraction 
$
(gT_\omega^{-1})^*=V_\omega(gP_1+P_2)
$
 yields the last statement. \ep
\end{proof}
\begin{rem}\label{err} The operator $T_\omega$ is cnu if and only if $0\leq g<1$, $|\alpha|<1$ and $|\delta|<1$. Moreover, in case  (\ref{speci}) holds, the corresponding random quantum walk operator $U_\omega(C)$ is purely ac by a general argument, see eq. (66) \S 5.4 of \cite{HJ}.
\end{rem}
The fact that $T_\omega$ is completely non-unitary has immediate consequences on the spectrum of  $U_\omega(C)$. In particular, the following result extends the description of the spectral diagram discussed in paragraph 5.6 of \cite{HJ}.
 \begin{prop}\label{Hsing}
If $0\leq g<1$, then  
$
\sigma(U_\omega(C))=\sigma_{ac}(U_\omega(C)),
$
 for all $\omega\in \Omega$.
\end{prop}
\begin{proof} We drop the dependence on $\omega$ and $C$ in the notation for this proof, for simplicity.
By Lemma \ref{lcnu}, we can assume $T$ is completely non-unitary. Let $P_{sing}$ be the spectral  projection onto the subspace $\cH^{sing}=\cH^{pp}(U)\cap \cH^{sc}(U)$ and recall that $P_0$ is the orthogonal projection onto $\cH_0$.  We first show that the subspace 
$\cH_0\cap \cH^{sing}$ reduces the operator $U$.
Let  $\psi \in \cH_0\cap \cH^{sing}$,
\be
U\psi=U  P_{sing} \psi=P_{sing} U\psi=P_{sing}\big(P_0 U\psi+(\I-P_0)U\psi\big),
\ee
where $(\I-P_0)U\psi\in\cH_b$, see (\ref{hbdef}). Using $P_{sing}\cH_b=0$, we get that
$U\psi=P_{sing}P_0U\psi.$
But then $\|U\psi\| \leq \|P_0 U\psi\|\leq \|U\psi\|$ implies $U\psi=P_0U\psi=P_0P_{sing}U\psi$ as well. Hence $\cH_0\cap \cH^{sing}$ is invariant under $U$. By a similar argument, this subspace is invariant under $U^*$ as well. Consequently, $\cH^{sing}$ reduces $T=P_0 U|_{\cH_0}$, which shows that $\cH^{sing}\cap \cH_0=\{0\}$ since $T$ is cnu and $g<1$. Repeating the argument with $\cH_0$ replaced by the horizontal subspace attached to $x\in \cT_4$ arbitrary eventually yields $\cH^{sing}=\{0\}$. \ep
\end{proof} 
\begin{rem} 
In view of Lemma \ref{gmg}, one sees that Lemma \ref{lcnu} and Proposition \ref{Hsing} carry over to the cases described in Remark \ref{highdim}, in case $T_\omega$ is cnu..
\end{rem}
 \subsection{ Extensions to Further Contractive Quantum Walks}\label{sextension}

 We make use of a symmetry of the contractive quantum walk  $T_\omega=\D_\omega^0 T$ with $T$ given by (\ref{quasiwalk}) in order relate it to $\widetilde T_\omega$ given by (\ref{matrixtt}). 
 Let 
 \bea 
 \cH_{\bf e}&=&\overline{\mbox{span}}\{a^m\otimes\tau, \ \ m\in 2\Z, \tau\in\{\pm1\}\}, \nonumber \\
  \cH_{\bf o}&=&\overline{\mbox{span}}\{a^m\otimes\tau, \ \ m\in 2\Z+1, \tau\in\{\pm1\}\}
\eea
denote the supplementary subspaces of $\cH_0$ consisting in even and odd sites only in configuration space. The definition  (\ref{quasiwalk}) of $T$ makes it clear that $T \cH_{\bf e}\subset \cH_{\bf o}$ and $T \cH_{\bf o}\subset \cH_{\bf e}$, and since $\D^0_\omega$ is diagonal,  the same is true for $T_\omega$. Therefore $\cH_{\bf e}$ is invariant under $T^2_\omega$ and by Lemma 2 in \cite{CD}, $\sigma (T^2_\omega)\setminus \{0\}=\sigma(T^2_\omega |_{\cH_{\bf e}})\setminus \{0\}$. Actually we have
\begin{prop} 
For all $0\leq g \leq 1$, and with definitions (\ref{matrixt1}) and (\ref{matrixtt}), 
 \be \widetilde T_\omega \simeq T^2_\omega|_{\cH_{\bf e}} \ \Rightarrow \
\sigma(\widetilde T_\omega)=\sigma(T^2_\omega).
\ee
Moreover,
\be\label{sfi}
\widetilde T_\omega = \bigoplus_{k\in \Z} S_\omega({2k+1}) \bigoplus_{k\in \Z} S_\omega({2k})
\ee
where,  for all $k\in \Z$, we have in the basis $\{e_{2k}, e_{2k+1}\}$, resp. $\{e_{2k+1}, e_{2k+2}\}$
\be
S_\omega({2k})=\mbox{\em diag}(e^{i\omega_{4k-1}}, e^{i\omega_{4k+2}})\begin{pmatrix}
\gamma & \delta \cr
\alpha & \beta
\end{pmatrix}, \mbox{\em resp.} \  
S_\omega({2k+1})=\mbox{\em diag}(e^{i\omega_{4k+1}}, e^{i\omega_{4k+4}})\begin{pmatrix}
\gamma & \delta \cr
\alpha & \beta
\end{pmatrix}.
\ee
\end{prop}
\begin{proof}
With the convention (\ref{rule}), $\cH_{\bf e}$ is spanned by $\{e_{4k}, e_{4k+1}, \ k\in \Z\}$. Relabelling these basis vectors according to $e_{4k}\mapsto e_{2k}$, $e_{4k+1}\mapsto e_{2k+1}$,  explicit computations yield $\widetilde T_\omega \simeq T^2_\omega|_{\cH_{\bf e}}$, as well as (\ref{sfi}). Observe that $ g\neq 0 $ iff $\widetilde T_\omega$ and $T_\omega$ are boundedly invertible and that if $g=0$, we have $0\in \sigma (\widetilde T_\omega)\cap  \sigma (T_\omega^2) $.  This yields isospectrality of $T_\omega^2$ and $\widetilde T_\omega$. \ep
\end{proof}
\begin{rems}
i) The restriction $T^2_\omega|_{\cH_{\bf o}}$ has an explicit form similar to $\widetilde T_\omega$ given by the composition (\ref{sfi}) in the reversed order.\\
ii) In particular, we deduce from the above that  $\widetilde T_\omega$ is unitary iff $g=1$, and that it is pure point  for $\beta\gamma\neq 0$, whereas it is absolutely continuous if $\beta=\gamma= 0$, \cite{JM}. \\
iii) All the spectral results we derive for $T_\omega$ hold for $\widetilde T_\omega$ via the spectral mapping theorem.
\end{rems}
\section{Spectral Analysis of $T_\omega$}

We use the following notations: $\sigma_p(A)$ denotes the set of eigenvalues of a bounded operator $A$ on $\cH$ and  $\sigma_{app}(A)$ denotes its approximate point spectrum. By definition, $\lambda\in \sigma_{app}(A)$ if and only if there exists a sequence of normalized vectors $\{\ffi_n\}_{n\in\N}$ such that $A\ffi_n-\lambda\ffi_n\ra 0$, as $n\ra\infty$. Recall that  $\sigma_p(A)\subset \sigma_{app}(A)$ and $\sigma(A)=\sigma_{app}(A)\cup {\overline{\sigma_p(A^*)}}$, where $\overline{X}=\{\bar x ,\ | \ x\in X\}$, for any $X\subset \C$. Also, $\sigma_{app}(A)$ is a nonempty closed set of $\C$ such that $\partial \sigma(A)\subset \sigma_{app}(A)$ and one has the disjoint union $\sigma(A)=\sigma_{app}(A)\cup {\overline{\sigma_{p_1}(A^*)}}$,  where ${\sigma_{p_1}(A^*)}=\{\lambda \in \C \ | \ \mbox{s.t.} \ \ker (A^*-\lambda)\neq \{0\}\ \mbox{and}\ \mbox{Ran}(A^*-\lambda)=\cH\}$ is open in $\C$, see \cite{Ku}.

The starting point of analysis of the contraction $T_\omega$ is Theorem 4.4 showing that $T_\omega$ admits a polar decomposition the components of which are bounded normal operators. We are thus naturally lead to the study of spectral properties of products of such operators. The only general result we are aware of in this direction, \cite{W}, provides estimates on the position of the spectrum of such products in terms of the numerical ranges of the components, which is however not strong enough for our purpose. We will use instead

\begin{thm}\label{gensigvk} 
Let $T=AB$, where $A$, $B$ are bounded normal operators on $\cH_0$ and let $B_c(r)$ denote the open disc of radius $r>0$ and center $c\in \C$. Then,
\bea \label{gensetres}
 B^{-1}\in \cB(\cH_0)&\Rightarrow & \bigcup_{\tau\in \rho(A)}\bigcap_{b \in \sigma(B)}B_{\tau b}(|b|\; \mbox{\em dist}(\tau, \sigma(A)) )\subset \rho(AB),
 \nonumber \\
A^{-1}\in \cB(\cH_0)&\Rightarrow &\bigcup_{\tau\in \rho(B)}\bigcap_{a \in \sigma(A)}B_{\tau a}(|a|\; \mbox{\em dist}(\tau, \sigma(B)) )\subset \rho(AB). 
\eea
\end{thm}
\begin{proof}
Under our assumption on $\tau$, and since $B$ is invertible, we have
\be
T-z=(A-\tau)B+\tau B-z=(A-\tau)\left(\un + (A-\tau)^{-1}(\tau B-z)B^{-1}\right)B,
\ee
which shows that $T-z$ is boundedly invertible if $\|(A-\tau)^{-1}(\tau B-z)B^{-1}\|<1$, thanks to Neumann's series. By the spectral theorem for normal operators applied to the continuous function $x\ra \frac{|\tau x -z|}{|x|}$ defined on the compact set $\sigma (B)$, and using $\|(A-\tau)^{-1}\|=1/\mbox{ dist}(\tau, \sigma(A))$, this condition is met if
\be
\max_{b\in \sigma(B)}\frac{|z-\tau b|}{|b|}<\mbox{ dist}(\tau, \sigma(A)).
\ee 
Therefore, given $\tau\in \rho(A)$, if $z\in \bigcap_{b\in \sigma(B)}B_{\tau b}(b\, \mbox{dist}(\tau, \sigma(A)) )$, then $z\in\rho(AB)$. Taking the union over $\tau \in \rho(A)$ yields (\ref{gensetres}).  The second inclusion is proven analogously, using $A$ invertible and identity for $\tau\in \rho(B)$
 \be
T-z=A(B-\tau)+\tau A-z=A\left(\un + A^{-1}(\tau A-z)(B-\tau)^{-1}\right)(B-\tau). 
\ee 
\ep
\end{proof}
\begin{rem}\label{discrho}
In case $A$ and $B$ have bounded inverses, we get for $\tau=0$ that $B_0(r_{AB})\subset \rho(AB)$, where $r_{AB}=\mbox{\em dist}(0, \sigma(A))\mbox{\em dist}(o, \sigma(B))>0$.  
\end{rem}
Applied to our case $T=VK$ with $\sigma(K)=\{g,1\}$, $0< g<1$, 
(\ref{gensetres}) simplifies and yields more specific estimates on $\rho(T)$ as a function of the spectrum of the unitary operator $V$.
\begin{cor}\label{simplif} Let $T=VK$ with $V$ unitary and $0<K=(P_1+gP_2)$, $0< g<1$.  Then
\bea\label{setres}
&&\bigcup_{\tau\in \rho(V)}B_{\tau g}(g\, \mbox{\em dist}(\tau, \sigma(V)) )\cap B_{\tau }(\mbox{\em dist}(\tau, \sigma(V)) )\subset \rho(T),\\
&&\label{cor2}
\bigcup_{\tau\in \rho(K)}\bigcap_{v\in \sigma(V)}B_{\tau v}(\mbox{\em dist}(\tau, \sigma(K)) )\subset \rho(T).
\eea
In particular,  \vspace{-.5cm}
\be
B_0(g)\subset \rho(T). \label{87}
\ee
Moreover, assume the arc $(-\theta,\theta)$ belongs to $\rho(V)$, with $0< \theta <\pi$. Then,
\bea\label{symset}
&&\bigcup_{\tau\in \R_+ \atop \alpha \in [-\theta, \theta] }B_{e^{i\alpha}\tau}(d_{e^{i\alpha}\tau})\cap B_{ge^{i\alpha}\tau}(gd_{e^{i\alpha}\tau} )
 \subset \rho(T),\\ \label{symset2}
 &&\bigcup_{\tau\in \R_- \atop \alpha \in [-\pi/2, \pi/2] }\bigcap_{e^{i\nu}\in\sigma(e^{i\alpha}V)}B_{e^{i\nu}\tau}(\delta_{e^{i\alpha}\tau})  \subset \rho(T), \ \ \mbox{where }\\
\label{distalph}
d_{e^{\pm i\alpha}\tau} &=& \mbox{\em dist}(e^{\pm i\alpha}\tau, \sigma(V)) = \sqrt{\tau^2-2\tau\cos(\theta-\alpha)+1} \ \mbox{with } \ \tau>0, \ \alpha \in [0, \theta], \\
\label{deltalph}
\delta_{e^{\pm i\alpha}\tau}&=& \mbox{\em dist}(e^{\pm i\alpha}\tau, \sigma(K))= \sqrt{\tau^2+2|\tau|g\cos(\alpha)+g^2} \ \mbox{with } \ \tau<0, \ \alpha \in [0, \pi/2].
\eea
\end{cor}

\begin{rems}\label{remset} i) The points $\tau\in \rho(V)$ in (\ref{setres}) such that $\mbox{dist}(\tau, \sigma(V))=|1-\tau|$ do not yield more information than (\ref{87}): $\tau<1$ implies $\bigcap_{k\in \sigma(K)}B_{\tau k}(k (1-|\tau|) )\subset B_0(g)$ and $\tau>1$ implies $\bigcap_{k\in \sigma(K)}B_{\tau k}(k (|\tau|-1) )\subset \C\setminus \overline{B_0(1)}$. This is the case when $\sigma(V)=\Ss$.\\
ii) At the expense of a rotation, we can associate to any arc in $\rho(V)$ two sets (\ref{symset}) and (\ref{symset2}) that belong to $\rho(T)$. The corresponding sets are both symmetrical with respect to the bisector of that arc. \\
iii)  Lemma \ref{sc} or 
 Remark \ref{remopt} shows that (\ref{87}) is optimal.
\end{rems}
\begin{proof} 
The first statements are mere rewritings of (\ref{gensetres}) and 
Remark \ref{discrho} implies (\ref{87}). For (\ref{symset}), we note that $w\in \C$ is such that $\mbox{dist}(w, \sigma(V))=|w-e^{\pm i\theta}|$ if $w=\tau e^{\pm i\alpha}$, with $\alpha \in [0, \theta]$ and $\tau\geq 0$, which establishes (\ref{distalph}). Whereas for (\ref{symset2}),  $w=-|\tau|e^{\pm i\alpha}$ with $\alpha \in [0, \pi/2]$  satisfies $\mbox{dist}(w, \sigma(K))=|w-g|=||\tau|e^{\pm i\alpha}+g|$ which yields (\ref{deltalph}). Then a change of variables allows us to express  (\ref{cor2}) as (\ref{symset2}) under our assumptions.
\ep
\end{proof}

Without attempting to provide a complete analysis, we describe (\ref{symset}) and  (\ref{symset2}) in some more details and show that (\ref{symset2}) provides less information in case $\sigma(V)$ displays one gap only. The proofs of the statements are provided in an Appendix. Let $C_c(r)$ denote the circle of center $c\in \C$ and radius $r>0$ and $\partial S$ denote the boundary of a set $S$. 
First consider (\ref{symset}) for $\alpha=0$. 
Because the intersection of discs can be non-empty when the intersection of their boundary is empty, there is a difference between (\ref{symset}) and the set $D(\theta)$ such that 
\be\label{deltad}
\partial D(\theta)=\bigcup_{\tau\in \R_+ }C_{\tau}(d_{\tau})\cap C_{g\tau}(gd_{\tau} ),
\ee 
and $D(\theta)$ contains the vertical segment between the intersection of two circles. We also set 
$R_\gamma(\theta)=\{z\in \C \, | \, \Re z>\gamma\cos(\theta)\}$.
\begin{lem}\label{form} With the notations above, and assuming $\alpha=0$, the LHS of (\ref{symset}) is given by
 \be\label{set1}
\bigcup_{\tau\in \R_+ }B_{\tau}(d_{\tau})\cap B_{g\tau}(gd_{\tau} )=D(\theta)\cup B_0(g)\cup R_1(\theta),  \ \  \mbox{for $\theta\in ]0,\pi/2[$ },
\ee
see Fig. \ref{gaps}, where
$\partial D(\theta)$ is given by the cubic curve
\bea\label{cubicurve}
y^2&=&\frac{x(x^2-x(1+g)\cos(\theta) +g)}{(1+g)\cos(\theta)-x} \ \mbox{with} \nonumber\\  \label{xtau}
x&=&-\frac{1+g}{2\tau}+(1+g)\cos(\theta)\in [0, (1+g)\cos(\theta)[, \ \ \mbox{for}\ \tau\in [1/(2\cos(\theta)), \infty[.
\eea
For $\pi/2\leq \theta<\pi$, 
\be\label{set11}
\bigcup_{\tau\in \R_+ }B_{\tau}(d_{\tau})\cap B_{g\tau}(gd_{\tau} )= B_0(g)\cup R_g(\theta).
\ee
Moreover, for fixed $0<\alpha<\theta$, assuming $0<\theta<\pi$,  we have
\be\label{alphaind}
\bigcup_{\tau\in \R_+ }B_{e^{ i\alpha}\tau}(d_{e^{ i\alpha}\tau})\cap B_{ge^{ i\alpha}\tau}(gd_{e^{ i\alpha}\tau} )\subset\bigcup_{\tau\in \R_+ }B_{\tau}(d_{\tau})\cap B_{g\tau}(gd_{\tau} ).
\ee
\end{lem}
\begin{figure}[htbp]
\begin{center}
      \includegraphics[scale=.21]{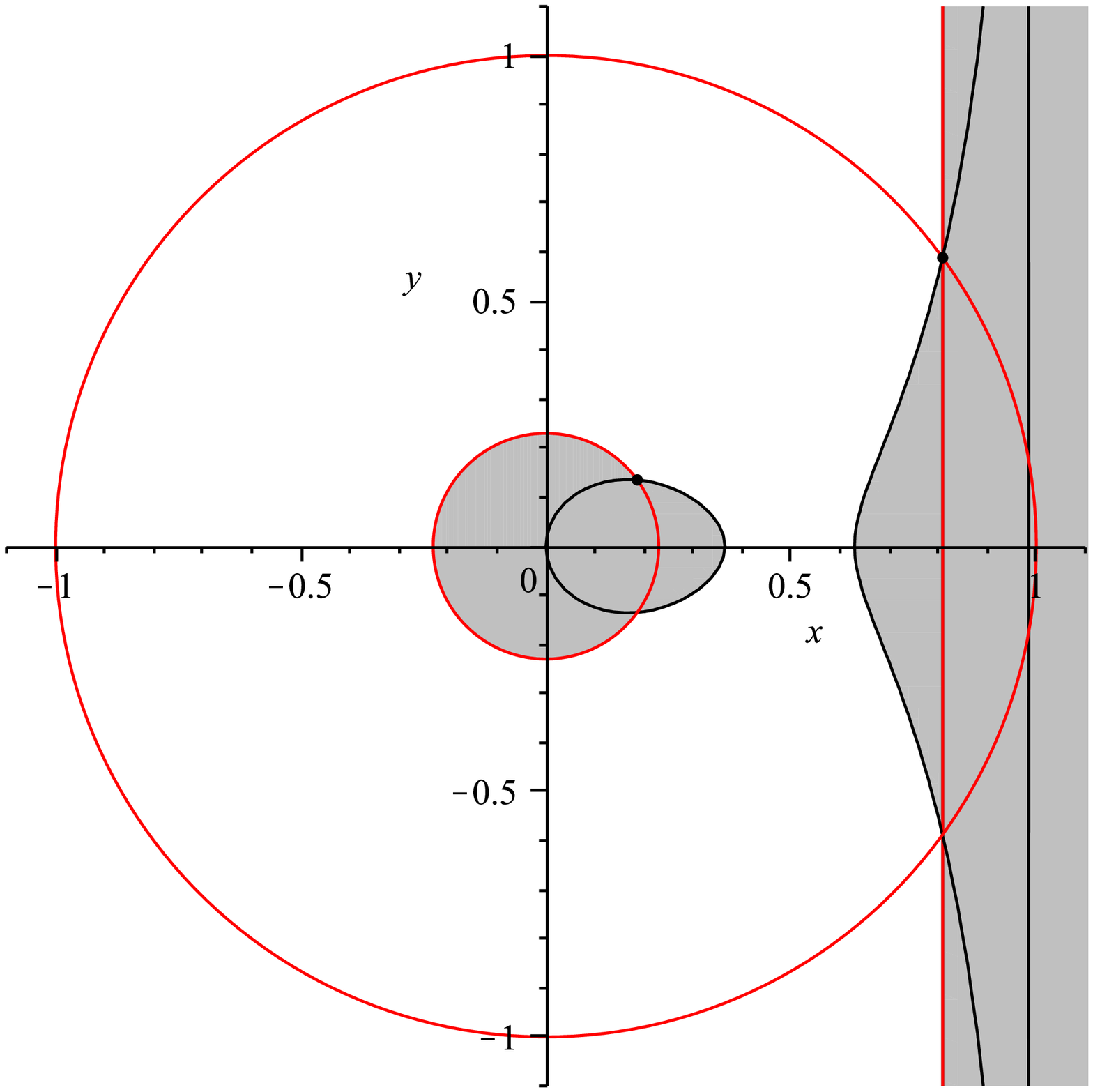} \includegraphics[scale=.21]{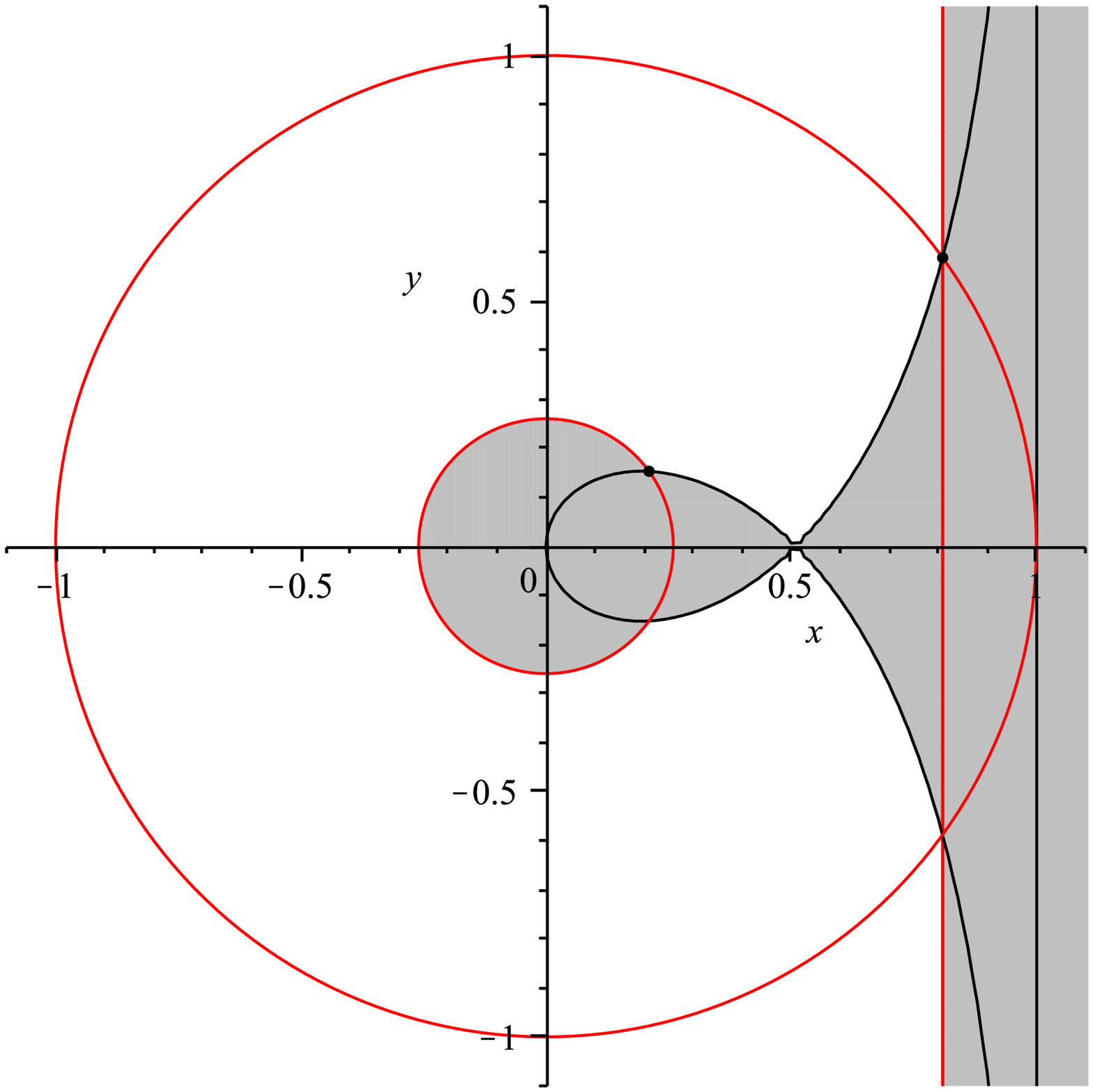} \includegraphics[scale=.21]{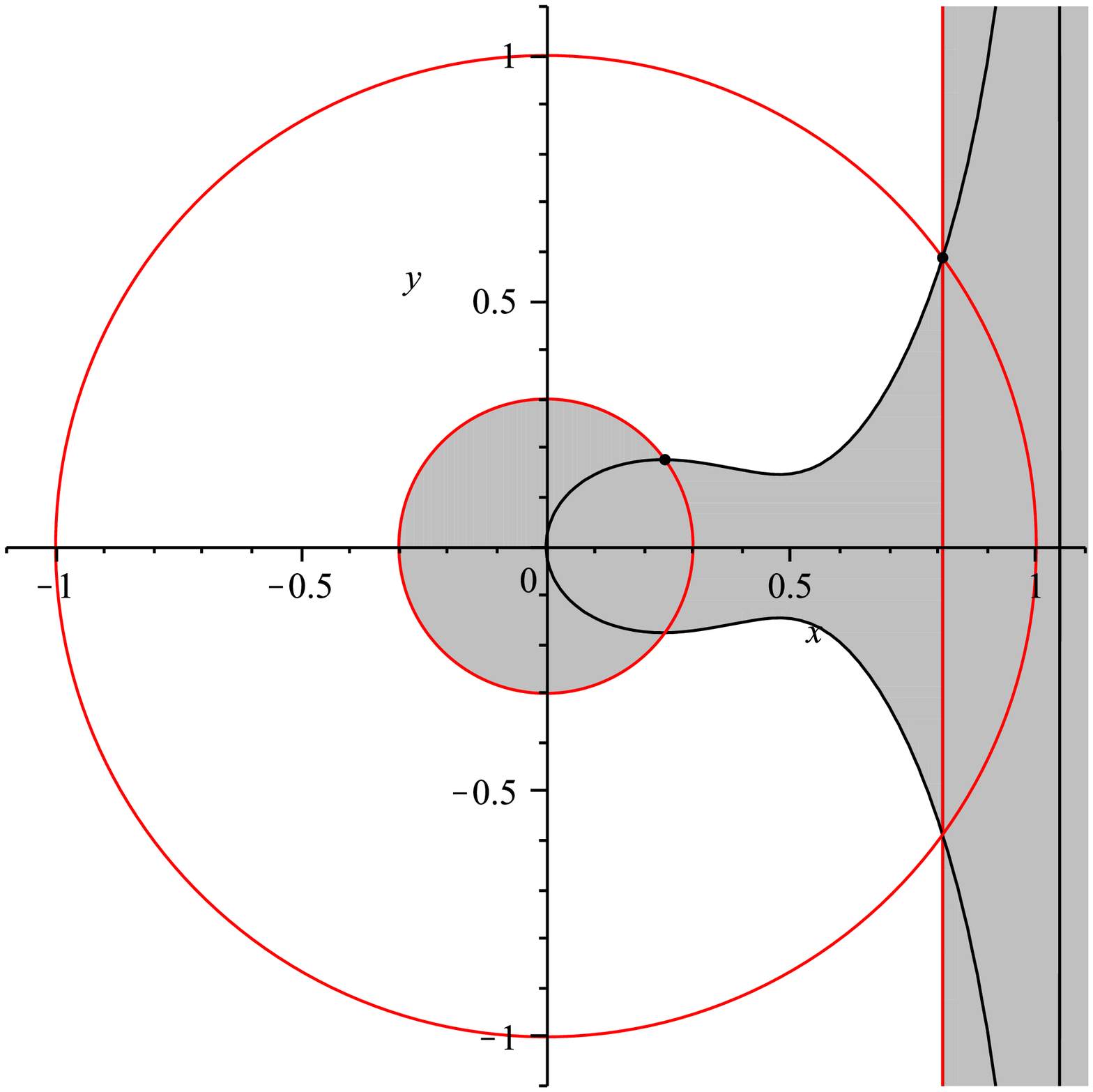}
   \end{center}\vspace{-.5cm}
   \caption{\footnotesize The sets $D\cup B_0(g)\cup R_1(\theta)$ for $0<\theta<\pi/2$ fixed and increasing values of $g$. The unit circle $\mathbb S$ and $g\mathbb S$ are indicated in red, whereas the black curves denote $\partial D$. The vertical red line corresponds to $\partial R_1(\theta)$.}\label{gaps}
\end{figure}
\begin{rems}\label{alpha1} i) In particular, under our assumptions,  the segment $[0,1]\subset \rho(VK)$ if $\cos^2(\theta)<\frac{4g}{(1+g)^2}\in ]0,1[$, see Figure \ref{gaps}.That this condition is necessary in general can be seen on the matrix case 
\be
V=\begin{pmatrix}e^{i\theta} & 0\cr 0& e^{-i\theta} \end{pmatrix}, \ K=\frac12\begin{pmatrix}1+g & 1-g \cr 1-g & 1+g \end{pmatrix}
\ee
such that $\sigma(VK)=\{\frac12(\cos(\theta)(1+g)\pm\sqrt{\cos^2(\theta)(1+g)^2-4g}) \}\subset \R_+^*$, if $\cos^2(\theta)\geq \frac{4g}{(1+g)^2}$.
\\
ii) The points $0, ge^{i\theta}$ and $e^{i\theta}$ belong to $\partial D(\theta)$ and correspond to the values of $\tau$ given by $1/(2\cos(\theta)), (1+g)/(2\cos(\theta))$ and $(1+g)/(2g\cos(\theta))$ respectively.
\end{rems}

To discuss the set (\ref{symset2}), we need some notations. For $\rho,\rho'>0$, we define, see Figure \ref{gamma},
\be
\Gamma_{\rho,\rho'}(\theta)=(B_{-e^{+i\theta}\rho}(\rho+\rho')\cap B_{-e^{-i\theta}\rho}(\rho+\rho')\cap R_{\rho'}(\theta))\cup B_0(\rho').
\ee
\begin{figure}[htbp]
\begin{center}
      \includegraphics[scale=.2]{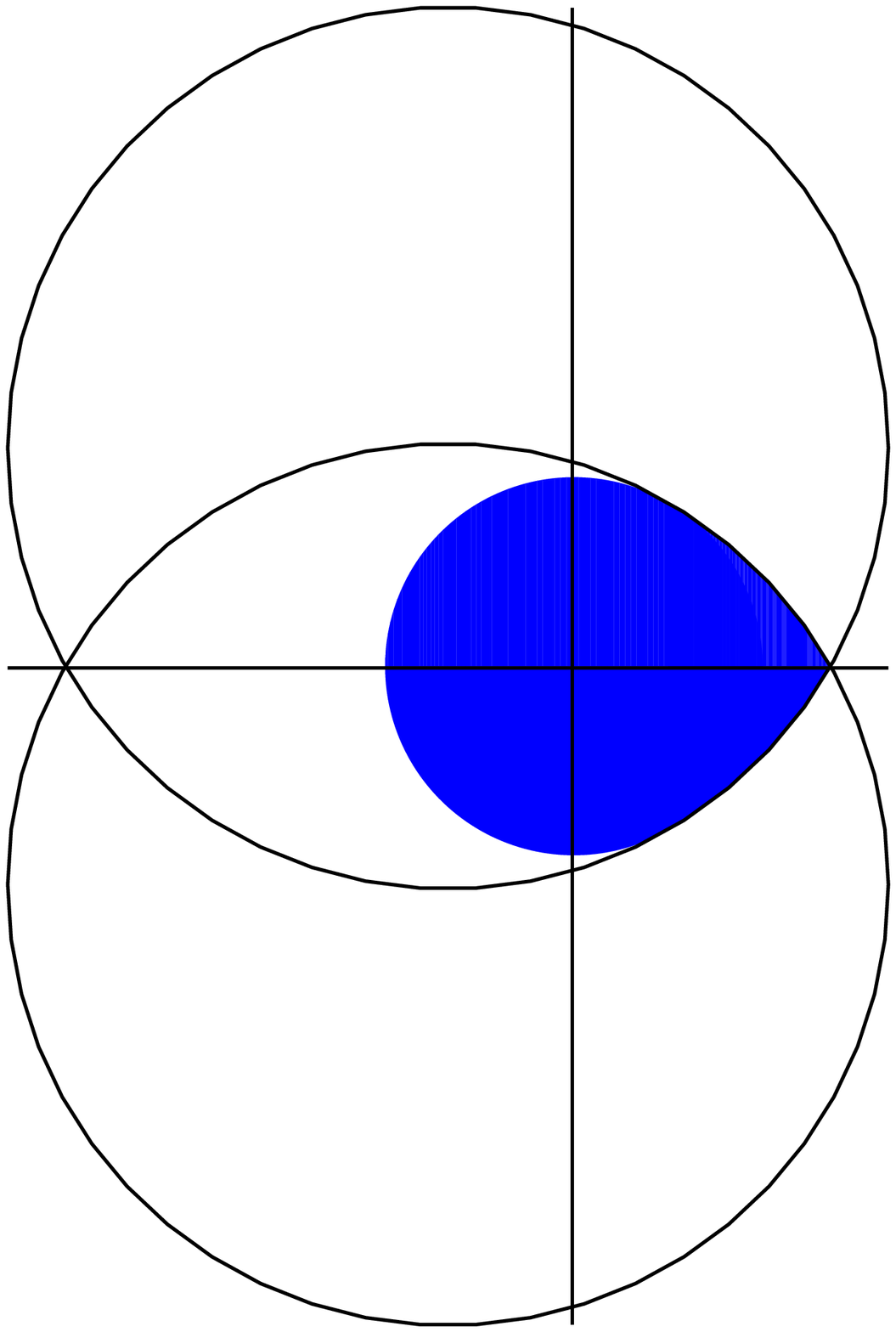}
   \end{center}\vspace{-.5cm}
   \caption{\footnotesize The set $\Gamma_{\rho,\rho'}(\theta)$.}\label{gamma}
\end{figure}
where the two discs $B_{-e^{\pm i\theta}\rho}(\rho+\rho')$  tangent to $B_0(\rho')$ at $\rho'e^{\pm i\theta}$.
We prove the following in an Appendix.
\begin{lem}\label{form2} Assume $\sigma(V)=\{e^{i\nu} \ \mbox{s.t.} \ \nu \in [\theta,\pi]\cup[-\pi,-\theta] \}$, with $\theta\in ]0,\pi[$. We have
\be\label{set2}
\bigcup_{\tau\in \R_-}\bigcap_{e^{i\nu}\in\sigma(V)}B_{e^{i\nu}\tau}(\delta_{\tau})  = B_0(g)\cup \Delta_g(\theta),
\ee
where $\Delta_g(\theta)$ denotes either the triangle defined by the points $ge^{i\theta}, ge^{-i\theta}, g/\cos(\theta)$  whenever $\theta<\pi/2$,  or $\Delta_g(\theta)$ denotes the set delimited by the two non-vertical lines passing by these points and the condition $\Re z\geq g\cos(\theta)$  whenver $\theta\in [\pi/2, \pi[$. 
Then, for each $\alpha\in ]0,\pi/2[$ fixed, 
\bea\label{reduct}
\bigcup_{\tau \in \R_-}\bigcap_{e^{i\nu}\in\sigma(V)}B_{e^{i\nu}e^{i\alpha}\tau}(\delta_{e^{i\alpha}\tau})  &=& e^{i\alpha}\bigcup_{|\tau| \in \R_+}\Gamma_{|\tau| , \delta_{\tau e^{ i\alpha}}-|\tau|}(\theta).
\eea
For any $\theta \in [0,\pi[$,  and all $\alpha\in ]0,\pi/2[$,
 \be\label{included}
 \bigcup_{\tau \in \R_-}\bigcap_{e^{i\nu}\in\sigma(V)}B_{e^{i\nu}e^{i\alpha}\tau}(\delta_{e^{i\alpha}\tau}) \subset \bigcup_{\tau\in \R_+ }B_{\tau}(d_{\tau})\cap B_{g\tau}(gd_{\tau} ).
 \ee
\end{lem}
\begin{exple}\label{ex1}{\em
Let us  illustrate the use of Theorem \ref{gensigvk}. Consider
\be\label{exple}
\tilde C(\xi,\eta)=
\begin{pmatrix} \cos(\eta) & \cos(\xi)\sin(\eta) &-\sin(\xi)\sin(\eta) \cr
0 &\sin(\xi) &  \cos(\xi)\cr
\sin(\eta) & -\cos(\xi)\cos(\eta)  &  \sin(\xi)\cos(\eta) \end{pmatrix}\in O(3), \ \ \ 
 \xi, \eta\in [0,\pi/2],
\ee
where $(\xi,\eta)$ is restricted to $[0,\pi/2]^2$ for simplicity. We thus compute that 
\bea
T, \ \mbox{resp.} \  V, \  \mbox{is characterized by} \ \begin{pmatrix} \cos(\eta) & -\sin(\eta)\sin(\xi) \cr
\sin(\eta) & \cos(\eta)\sin(\xi)  \end{pmatrix},
\mbox{resp.} \ \begin{pmatrix} \cos(\eta) & -\sin(\eta) \cr
\sin(\eta) & \cos(\eta)  \end{pmatrix}.
\eea
Moreover, Fourier methods yield 
\be
\sigma(V)=\{z\in \Ss \ | \ \arg z \in [\eta, \pi-\eta]\cup [-\pi+\eta, -\eta]\}.
\ee
Assuming the common distribution $d\nu$ of phases has support given by 
\be\label{suppdnu}
\mbox{supp }d\nu=[-\eps, \eps], \ \mbox{with}\  \eps<\eta,
\ee we have thanks to the general almost sure relation $\sigma(V_\omega)=\sigma(V)e^{i\, \mbox{\scriptsize supp}(d\nu)}$ which holds for products of unitary operators of that sort, see Section 5.1 of \cite{J1}, for example,
\be
\sigma(V_\omega)=\{z\in \Ss \ | \ \arg z \in [\eta-\eps, \pi-\eta+\eps]\cup [-\pi+\eta-\eps, -\eta+\eps]\}, \ \mbox{a.s.}
\ee
Hence, Corollary \ref{simplif} applies with $\theta=\eta-\eps$ and $g=\sin(\xi)$, and gives rise to two regions of $\rho(T_\omega)$: one described in Lemma \ref{form}, and its symmetric image with respect to the vertical axis. In particular, the spectrum of the corresponding $T_\omega$ is separated into two disjoint parts if 
\be\label{splits}
\cos^2(\eta-\eps)\leq\frac{4\sin(\xi)}{(1+\sin(\xi))^2}.
\ee
}
\end{exple}

Let us continue with some general links between the spectral properties of $T_\omega$ and $U_\omega(C)$.
\begin{lem} \label{kerzero}
Let $U$ be unitary on $\cH$ and $P_0$ be an orthogonal projector. For any $\ffi\in \cH$
\bea\label{evpu}
UP_0\ffi&=&e^{i\theta}\ffi \Rightarrow \ffi=P_0\ffi \ \mbox{and }\ e^{i\theta}\ffi=U\ffi=P_0UP_0 \ffi,
\\
\label{evpup}
P_0U\ffi&=&e^{i\theta}\ffi \Rightarrow \ffi=P_0\ffi \ \mbox{and }\ e^{i\theta}\ffi=U\ffi=P_0UP_0 \ffi.
\eea
Moreover, writing $Q_0=\un -P_0$, we get
\be
\ker  Q_0UP_0 =\{0\} \Rightarrow \sigma_p(UP_0)\cap \Ss =\sigma_p(P_0U)\cap \Ss=\sigma_p(P_0UP_0)\cap \Ss=\emptyset.
\ee 
Furthermore, let $T=P_0U P_0|_{P_0\cH}$. If $e^{i\theta}\in \sigma_{app}(T)\setminus \sigma_p(T)$, then $e^{i\theta}\in \sigma_{app}(U)$.
\end{lem}
\begin{proof} Taking the norm of the left hand side of (\ref{evpu}) yields  
$P_0\ffi=\ffi$, $Q_0UP_0\ffi=0$ and the first identities follow.  For (\ref{evpup}),
$P_0U\ffi=e^{i\theta}\ffi =P_0e^{i\theta}\ffi$ gives the results directly. Now,
$
P_0U\ffi=e^{i\theta}\ffi \ \Leftrightarrow \ UP_0\psi=e^{i\theta}\psi  \ \mbox{where} \ \psi=U\ffi
$
shows with  (\ref{evpu}) that (\ref{evpup}) implies $Q_0UP_0\psi=0$. Similarly, $P_0UP_0\ffi=e^{i\theta}\ffi$ implies $Q_0UP_0\ffi=0$. Thus, if $\ker Q_0UP_0 =\{0\} $, we get the absence of eigenvalue of modulus one for $UP_0$, $P_0U$ and $P_0UP_0$. 
Finally, let $e^{i\theta}\in \sigma_{app}(T)\setminus \sigma_p(T)$ and $\ffi_n\in P_0\cH$ s.t. $\| \ffi_n \|=1$ and $T\ffi_n-e^{i\theta}\ffi_n\ra 0.$  By assumption, 
$
\|U\ffi_n\|^2=\|e^{i\theta}\ffi_n+(P_0U\ffi_n-e^{i\theta}\ffi_n)\|^2+\|Q_0U\ffi_n\|^2,
$ 
where the parenthesis in the right hand side tends to zero, as $n\ra \infty$. As $U$ is unitary, we have $\lim_{n\ra \infty}Q_0U\ffi_n=0$. Consequently,  $e^{i\theta}\in \sigma_{app}(U)$ since
$
U\ffi_n-e^{i\theta}\ffi_n = T\ffi_n-e^{i\theta}\ffi_n+Q_0U\ffi_n\ra 0, \ \mbox{as} \ n\ra \infty.
$
\ep
\end{proof}
\begin{rem} i) The same result holds with $T^*$ and $U^*$ in place of $T$ and $U$.\\
ii)  If $\ker (Q_0UP_0)=\{0\}$, $\lim_{n\ra \infty}Q_0U\ffi_n=0$ implies that the operator $[Q_0UP_0]^{-1}:\ran  Q_0UP_0\subset Q_0\cH\ra P_0\cH$ is not bounded. 
\end{rem}
Let us also recall the following properties.
\begin{lem}\label{easyspect} Let $T=V(P_1+gP_2)$ and
$\ffi \in \cH_0$ such that  $T \ffi =\lambda \ffi$. Then for all $0<g<1$,
\bea
|\lambda|=1& \Rightarrow & \begin{matrix} \ffi=P_1\ffi 
\ \mbox{and } \ V \ffi=P_1V P_1\ffi=\lambda \ffi \end{matrix},\nonumber \\
|\lambda|=g& \Rightarrow & \begin{matrix} \ffi=P_2\ffi 
\ \mbox{and } \ V \ffi =P_2V P_2 \ffi =(\lambda/g) \ffi \end{matrix}.
\eea
Consequently, \vspace{-.80cm}
\bea
\ker P_2V P_1 =\{0\} & \Rightarrow & \sigma_p(T)\cap \Ss =\emptyset, \ \mbox{and}  \nonumber \\  
\ker P_1V P_2 =\{0\} & \Rightarrow & \sigma_p(T)\cap g\Ss =\emptyset.
\eea
If $g=0$,  \vspace{-.80cm}
\bea
\sigma(T)=\sigma(P_1 V P_1|_{P_1\cH_0})\cup \{0\}.\label{88}
\eea
\end{lem}
\begin{proof} All statements except the last one are consequences of the proof of Lemma \ref{lcnu}.
If $g=0$, $T=V P_1$, so that $\ker \ T=P_2\cH_0$. Statement (\ref{88}) is a consequence of (\ref{scco}) and (\ref{critschur}) in the proof of Theorem \ref{suff} below. 
\ep
\end{proof}
\begin{rems} 
i) Analogous statements hold when $T$ is replaced by $(P_1+gP_2)V=V^*T V$. In particular, the results hold for $T^*$.
\end{rems}
Next, we come back to our random setting and make further use of the structure of $K$ to apply the Feschbach-Schur method in order to obtain conditions on the coefficients of $\tilde C$ (\ref{const}) that ensure that for all realizations $\omega\in \Omega$,  $\mbox{spr }(T_\omega)<\|T_\omega\|=1$, in case $g<1$.
\begin{thm} \label{suff} Let $T_\omega=V_\omega(P_1+gP_2)$, where $P_j$ are defined in (\ref{specdeck}) and $0\leq g<1$. Consider $P_jVP_k=V_{jk}$, $j,k\in \{1,2\}$, as operators on $P_k\cH$. If $\|V_{11}\|<1$, 
then, for all realizations $\omega\in\Omega$
\be\label{nonper}
g<\frac{1-\|V_{11} \|}{\|V_{21} \|\|V_{12} \|+\|V_{22} \|(1-\|V_{11} \|)}\ \ \Rightarrow \ \  \mbox{\em spr }(T_\omega)<1.
\ee
Moreover, the  set $\{|z|<g\}\cup\{r(V)<|z|\leq 1\}\subset \rho(T_\omega)$ for all $\omega\in\Omega$, where
\be\label{annulus}
r(V)=\frac{1}{2}\left(\|V_{11} \|+g\|V_{22}\| +\sqrt{(\|V_{11} \|-g\|V_{22} \|)^2+4g\|V_{21} \|\|V_{12} \|}\|\right).
\ee
\end{thm}
\begin{rems}
i) The result is deterministic and holds for any operator $T=V(P_1+gP_2)$, where $V$ is unitary and $\{P_j\}_{j=1,2}$ are supplementary orthogonal projectors. \\
ii) In case $V_\omega$ is given by Theorem \ref{t1},  (\ref{nonper}) yields a somehow implicit condition since the norms $\|V_{jk}\|$ depend on $g$, see Lemma \ref{qlem} and Example \ref{ex2} below. \\
iii) Remark \ref{remopt} below shows that $r(V)$ is optimal.\\
iv) This infinite dimensional result is reminiscent of the works \cite{WF, B}, which consider matrices of the form $T_\omega=V_\omega K$ where $V_\omega$ is a unitary, Haar distributed matrix and $K>0$ is given.  It is shown under various assumptions that a density of eigenvalues of $T_\omega$ can be defined, which is supported in a deterministic ring.
\end{rems}
\begin{proof} It is enough to prove the second statement. We start with the deterministic case.
Given $K=P_1+gP_2$, we split $\cH_0$ as $\cH_0=\cH_1\bigoplus \cH_2$ where $\cH_j=P_j\cH_0$. Writing $T=VK$ as a bloc structure according to this decomposition, we have for any $z\in \C$
\be
T-z\un=\begin{pmatrix}V_{11} -z\un_1& g V_{12} \cr
V_{21} & g V_{22}-z\un_2
\end{pmatrix},
\ee
where $\un_j=P_j|_{\cH_j}$ is the identity operator in $\cH_j$ and $V_{jk}=P_jVP_k$ are understood as operators from $\cH_k$ to $\cH_j$, $j,k\in\{1,2\}$.
For any $z\in \rho(g V_{22})$, we consider the Schur complement $F(z)\in \cB(\cH_1)$ defined by
\be\label{scco}
F(z)=(V_{11} -z\un_1)-gV_{12} (g V_{22}-z\un_2)^{-1}V_{21},
\ee
such that 
\be\label{critschur}
z\in \rho(T)\cap  \rho(g V_{22}) \Leftrightarrow 0\in \rho(F(z)).
\ee 
As $V$ is unitary, we have $g \|V_{22}\|\leq g<1$, so that $F: \{|z|>g\}\ra \cB(\cH_1)$ is well defined. 
If $z\in \rho(V_{11})\cap \Ss$, we can write
\be
F(z)=(V_{11} -z\un_1)\left(\un_1-g(V_{11} -z\un_1)^{-1}V_{12} (g V_{22}-z\un_2)^{-1}V_{21}\right),
\ee
which has a bounded inverse if
$
g\|(V_{11} -z\un_1)^{-1}V_{12} (g V_{22}-z\un_2)^{-1}V_{21}\|<1.
$
Assuming that $\|V_{11} \|<1$, we have $\{|z|>\|V_{11} \|\}\subset \rho(V_{11})$ and for $|z|>\max{(g \|V_{22} \|, \|V_{11} \|)}$, 
\bea
g\|(V_{11} -z\un_1)^{-1}V_{12} (g V_{22}-z\un_2)^{-1}V_{21}\| \leq
\frac{g\|V_{12} \|\|V_{21} \|}{(|z|-\|V_{11} \|)(|z|-g\|V_{22} \|)}.
\eea
The inner radius $r(V)$ of the ring (\ref{annulus}) is defined so that the right hand side above is strictly smaller than one and it satisfies $\max{(g\|V_{22} \|, \|V_{11} \|)}\leq r(V) < 1$ whenever
$g<\frac{1-\|V_{11} \|}{\|V_{21} \|\|V_{12} \|+\|V_{22} \|(1-\|V_{11} \|)}$. Thus, according to (\ref{critschur}), this implies that the ring (\ref{annulus}) belongs to the resolvent set of $T$, which yields the result for $T$ in place of $T_\omega$. 

To get the result for the random case with $V$ replaced by $V_\omega$, it  is enough to show that 
\be
\|P_jV_\omega P_k\|=\|P_j V P_k\|=\|V_{jk}\|, \ \ \forall \ j,k\in \{1,2\}. 
\ee
This is a consequence of the following lemma, which ends the proof of the theorem. \ep
\end{proof}

\begin{lem}\label{qlem} Let  $\{v_j^{(k)}\}_{k\in \Z}$
be the orthonormal basis of $\cH_j$, $j=1,2$ given by (\ref{evek}). Then
\bea�\label{vjk}
P_1V_\omega P_1v_1^{(k)}
&=&\frac{1}{1-g^2}\left(- e^{i\omega_{2k-1}}\bar q(s\gamma-q\delta)v_1^{(k-1)}+e^{i\omega_{2k+2}}\bar s(s\alpha-q\beta)v_1^{(k+1)}\right)\\ \nonumber
P_2V_\omega P_2v_2^{(k)}&=&\frac{1}{1-g^2}\left( -e^{i\omega_{2k-1}}st v_2^{(k-1)}-e^{i\omega_{2k+2}}qrv_2^{(k+1)}\right)\\ \nonumber
P_2V_\omega P_1 v_1^{(k)}&=&\frac{1}{1-g^2}\left(s(s\gamma-q\delta)e^{i\omega_{2k-1}}v_2^{(k-1)}+
q(s\alpha-q\beta)e^{i\omega_{2k+2}}v_2^{(k+1)}\right)\\ \nonumber
P_1V_\omega P_2 v_2^{(k)}&=&\frac{1}{1-g^2}\left( e^{i\omega_{2k-1}}\bar q t   v_1^{(k-1)}-\bar s r e^{i\omega_{2k+2}}v_1^{(k+1)} \right). \nonumber
\eea
Defining coefficients $w_\pm^{(ij)}$ by 
\be\label{coefvjk}
P_iV_\omega P_jv_j^{(k)}=e^{i\omega_{2k-1}}w_+^{(ij)}v_i^{(k-1)}+e^{i\omega_{2k+2}}w_-^{(ij)}v_i^{(k+1)},
\ee
we have \vspace{-.8cm}
\bea\label{normvjk}
\|P_jV_\omega P_k\|&=&|w_+^{(jk)}|+|w_-^{(jk)}|=\|V_{jk}\|
\eea
and, for all $i,j\in \{1,2\}$,
\be\label{eqkernul} 
\ker P_iV_\omega P_j \neq \{0\} \Leftrightarrow P_iV_\omega P_j=0.
\ee
Let  $\D_\eta^{(j)}$ and $\D_\xi^{(j)}$ be defined in
the orthonormal basis $\{v_j^{(k)}\}_{k\in \Z}$ of $\cH_j$ by
$
\D_\eta^{(j)}=\mbox{ \em diag }(e^{i\eta^{(j)}_k})$, and $\D_\xi^{(j)}=\mbox{\em diag }(e^{i\xi^{(j)}_k}),
$ 
where, for $p\geq 1$
\bea
\eta^{(j)}_{2p}&=&\sum_{l=0}^p \omega_{4l}-\sum_{l=0}^{p-1}\omega_{4l+1}, \ \ \ \
\eta^{(j)}_{2p+1}= \sum_{l=0}^p \omega_{4l+2}-\sum_{l=0}^{p-1}\omega_{4l+3}\\
\xi^{(j)}_{2p}&=&\sum_{l=0}^{p-1} \omega_{4l+3}-\sum_{l=0}^{p-1}\omega_{4l+2}, \ \
\xi^{(j)}_{2p+1}= \sum_{l=0}^p \omega_{4l+1}-\sum_{l=0}^{p}\omega_{4l},
\eea
and, for $p\leq 0$
\bea
\eta^{(j)}_{2p}&=&-\sum_{l=p+1}^1 \omega_{4l}+\sum_{l=p}^{1}\omega_{4l+1}, \ \ \ \
\eta^{(j)}_{2p+1}= -\sum_{l=p+1}^1 \omega_{4l+2}+\sum_{l=p}^{0}\omega_{4l+3}\\
\xi^{(j)}_{2p}&=&-\sum_{l=p}^{0} \omega_{4l+3}+\sum_{l=p}^{1}\omega_{4l+2}, \ \
\xi^{(j)}_{2p+1}= -\sum_{l=p+1}^1 \omega_{4l+1}+\sum_{l=p+1}^{1}\omega_{4l}.
\eea
Then, \vspace{-.55cm}
\be\label{rel}
P_jV_\omega P_k= \D_\eta^{(j)} V_{jk}  \D_\xi^{(k)} \simeq \D_\xi^{(k)} \D_\eta^{(j)} V_{jk}.
\ee
\end{lem}
\begin{proof} The expressions of $P_iV_\omega P_j$ in the bases $\{v_j^{(k)}\}_{k\in \Z}$ are obtained by 
explicit computations making use of (\ref{evek}), 
\be
e_{2k}=\frac{(\bar s v_1^{(k)}+q v_2^{(k)})}{\sqrt{|q|^2+|s|^2}}, \ \ e_{2k+1}=\frac{-\bar q v_1^{(k)}+s v_2^{(k)}}{\sqrt{|q|^2+|s|^2}},
\ee
and of the constraint (\ref{const}). Identity (\ref{normvjk}) is established by a classical argument and (\ref{eqkernul})  is a direct consequence of this identity. Relation (\ref{rel}) is also a matter of verification.
\ep
\end{proof}\\
\begin{rem} \label{nicexp} With $\det\begin{pmatrix}\alpha & \beta \cr \gamma & \delta \end{pmatrix}=ge^{i\chi}$, see Remark \ref{detg}, and constraint (\ref{const}), we have
\be
\|V_{11}\|=\frac{|\delta-\bar \alpha ge^{i\chi}|+|\alpha -\bar \delta ge^{i\chi}|}{1-g^2},
\ee
where the first / second term is the modulus of the coefficient of $v_1^{(k-1)}$ / $v_1^{(k+1)}$ in (\ref{coefvjk}).
\end{rem}

We establish further properties of $V_{jk}$ and $P_jV_\omega P_k$ as operators from $\cH_k$ to $\cH_j$, that we present in an abstract form.
\begin{prop}\label{spw}  Let $W$ be an operator that takes a  tridiagonal form in an orthonormal basis of $l^2(\Z)$ whose sole non zero coefficients satisfy 
\be
|W_{j, j+1}|=W_-, \ \mbox{and }\  |W_{j, j-1}|=W_+, \ \forall j\in \Z.
\ee 
Assume, without loss, that $W_+\geq W_->0$. Then, $\|W\|=W_++W_-$ and
\bea
\mbox{If \ } W_+\leq 1,&& \{|z|< W_+-W_-\}\subset \rho(W) \nonumber \\
\mbox{If \ } W_+ > 1,&&  \{|z|< (W_+-W_-)/(2W_+-1)\}\subset \rho(W).
\eea
If $W$ is further translation invariant, 
$
W_{j, j+1}=w_-, \ \mbox{and }\  W_{j, j-1}=w_+, \forall j\in \Z,
$ 
then $W$ is normal and
$\mbox{\em spr }(W)=\|W\|=|w_+|+|w_-|$.

\end{prop}
\begin{rem} 
The radius of both disks contained in $\rho(W)$ is smaller than one.
\end{rem}
\begin{proof} The norm of $W$ was already mentioned above.
The structure of $W$ is such that we can write $W=W^+ S_+ + W^- S_-$, where the non zero matrix elements of the operator $S_+/S_-$ lie on the diagonal immediately above/below the main diagonal, and all have modulus one;  $S_\pm$ are unitarily equivalent to standard shifts. Thus, for any $|z|\neq 1$, we can write
\bea
W-z&=&W^+ (S_+-z) + W^- S_- -z(1-W^+)\nonumber\\
&=&W^+ (S_+-z)\Big(\un + \frac{(S_+-z)^{-1}}{W^{+}}\Big(W^- S_- -z(1-W^+)\Big)\Big).
\eea
Since 
\be\label{rsha}
\left\|  \frac{(S_+-z)^{-1}}{W^{+}}\Big(W^- S_- -z(1-W^+)\Big)  \right\| \leq \frac{W^-+|z||1-W^+|}{W^+|1-|z||},
\ee
the Neumann series implies that $W-z$ admits a bounded inverse if the right hand side of (\ref{rsha}) is bounded above by one. Considering small values of $|z|$ and dealing with the different cases for $W^+$, we get the result.
In case $W$ is translation invariant, we obtain by Fourier methods that $W$ is unitarily equivalent to a scalar multiplication operator
\be
W\simeq W(x)=e^{ix}w_+ + e^{-ix}w_- \ \ \mbox{on $L^2(\T; \C)$}.
\ee
This operator is obviously normal, which ends the proof. 
\ep \end{proof}

Hence, the translation invariant contractions $P_jVP_j|_{\cH_j}=V_{jj}$ with  tri-diagonal representations  in the orthonormal basis of $\cH_j$ given by  $\{v_j^{(k)}\}_{k\in \Z}$, $j=1,2$, for $0\leq g<1$, with coefficients $w_\pm^{(jj)}$ defined by (\ref{coefvjk})
is normal and satisfies
$
\mbox{spr }(V_{jj})=\|V_{jj}\|=| w_+^{j(j)}|+|w_-^{(jj)}|.
$

\begin{exple}\label{ex2}
{\em
Let us apply the results above to Example \ref{ex1} where $\tilde C$ defined by equation (\ref{exple}). Recall that in this case $g=\sin(\xi)$, and $\xi,\eta\in [0,\pi/2]$. We get 
\be
\|V_{11}\|=\cos(\eta), \ \  \|V_{21}\|=\sin(\eta), \ \ \|V_{22}\|=\cos(\eta),\ \ \|V_{12}\|=\sin(\eta).
\ee
Thus, for $\eta,\xi \in ]0,\pi/2[$ so that $g>0$, $\|V_{11}\|<1$ and for $\xi$ small enough so that
\be\label{crxi}
\sin(\xi)<\frac{1-\cos(\eta)}{\sin^2(\eta)+\cos(\eta)(1-\cos(\eta))},
\ee
condition (\ref{nonper}) holds and we get
\be\label{rxi}
r(V)=\frac12\left(\cos(\eta)(1+\sin(\xi))+\sqrt{\cos^2(\eta)(1-\sin(\xi))^2+4\sin(\xi)\sin^2(\eta)}\right).
\ee
Actually, all corresponding operators $P_jV_\omega P_k$ in this case map the basis vector $v_k^{(n)}$ to one of $v_j^{(n\pm 1)}$ only.
In particular, $P_1V_\omega P_1|_{\cH_1}$ and $P_2V_\omega P_2|_{\cH_2}$ are unitarily equivalent to $\cos(\eta)S_1$ and $\cos(\eta)S_2$ respectively, where $S_j$ is the standard shift on $P_j\cH_j$. Hence, 
\be
\sigma(P_1V_\omega P_1|_{\cH_1})=\cos(\eta)\Ss \ \ \mbox{and} \ \
 \sigma(P_2V_\omega P_2|_{\cH_2})=\cos(\eta)\Ss.
\ee
Thus, assuming a phase distribution satisfying (\ref{suppdnu}) and parameters such that condition (\ref{crxi}) holds, we have excluded the presence of spectrum of the corresponding non-unitary operator $T_\omega$ in the union of the ring of inner radius (\ref{rxi}) and of the symmetric sets characterized by Lemma \ref{form}. Moreover, for suitable values of the parameters condition (\ref{splits}) holds as well and $\sigma(T_\omega)$ is contained in two disjoint sets separated by the real axis.
}
\end{exple}
The following more specific properties hold.
\begin{lem}\label{loffdiag} We have
\bea\label{offdiag}
\|V_{11}\|=0   \Leftrightarrow \ \|V_{22}\|=0 \ 
 \Leftrightarrow \  \tilde C\in \left\{ 
\begin{pmatrix}
0 & 0 & \beta  \cr
q & g & 0  \cr
\gamma & t & 0  
\end{pmatrix}, 
\begin{pmatrix}
0 & r & \beta  \cr
0 & g & s  \cr
\gamma & 0 & 0  
\end{pmatrix}
\right\}\subset U(3),
\eea
and, \vspace{-.8cm}
\bea\label{diag}
V_{jk}=0 \ \mbox{for some } k\neq j  &\Leftrightarrow &  V_{jj}\ \mbox{unitary for all $j\in\{1, 2\}$} \\  \nonumber 
\Leftrightarrow\ \mbox{$V_{jj}\simeq S_j$, $S_j$ a shift on $\cH_j$
}& \Leftrightarrow &
\tilde{C}\in \left\{ 
\begin{pmatrix}
\alpha & r & 0  \cr
q & g & 0  \cr
0 & 0 & \delta  
\end{pmatrix}, 
\begin{pmatrix}
\alpha & 0 & 0  \cr
0 & g & s  \cr
0 & t & \delta  
\end{pmatrix}
\right\}\subset U(3).
\eea
\end{lem}
\begin{rems} i) In case $V_\omega$ is off-diagonal with respect to $\cH_0=\cH_1\bigoplus\cH_2$, so that (\ref{offdiag}) and Lemma \ref{sc} hold, we saw that for all $0\leq g<1$ and all $\omega$, 
$\sigma(T_\omega)\subset \{z\in \C\ | |z|=\sqrt{g}\}$. We recover this result  by noting that  $V_\omega$ off-diagonal implies for $z\neq 0$
\be
F(z)=-z\left(\un_1-P_1V_\omega^2P_1g/z^2\right),
\ee
where $P_1V_\omega^2 P_1|_{\cH_1}$ is unitary. Hence $F(z)$ is boundedly invertible iff $z^2\in \sigma(gP_1V_\omega^2 P_1|_{\cH_1})$. 
\\
ii) In case $V_\omega$ is diagonal with respect to $\cH_0=\cH_1\bigoplus\cH_2$, so that (\ref{diag}) and Lemma \ref{sc} hold, we saw that for all $0\leq g<1$ and all $\omega$,
$\sigma(T_\omega)=\Ss\cup g\Ss. $
\end{rems}
\begin{proof}
The tridiagonal matrix representation of $V_{jj}$ stems from (\ref{vjk}), 
which yields the first statement.
The last statements are obtained by discussing the conditions $w_-^{(jj)}=w_+^{(jj)}=0$ 
depending on the fact that $q,s$ are zero or not. We first note that the condition $g<1$ forbids $q=s=0$ or $r=t=0$.
For $\|V_{11}\|=0$, the case $qs\neq 0$, is impossible: the expansion of $\det(\tilde C)$ with respect to the second column and $w_-=w_+=0$ imply $\det(\tilde C)=g(\alpha \delta-\gamma \beta)$, which is of modulus 1. This implies $g=|(\alpha \delta-\gamma \beta)|=1$ and $q=s=0$, a contradiction.  If $qs=0$, one gets that $\alpha$ or $\delta$ equals 1, which with condition
(\ref{lcm}) yield the result. Similarly, $\|V_{22}\|=0$ imply $q=t=0$ or $s=r=0$ and condition
(\ref{lcm}) again yields the result. The assertions regarding the off diagonal parts of $V_\omega$ are readily obtained by the same type of considerations and the fact that $V_\omega$ is unitary. \ep
\end{proof}
\subsection{Ergodicity}

We briefly recall here a spectral consequences of our hypothesis on the way the randomness enters the operator $T_\omega$. Ergodicity provides a tool to estimate from below the spectrum of $T_\omega$, almost surely. Our setup actually enters the more general theory of pseudo-ergodic operators, as developed in \cite{D1, D2}, of which ergodic operators are special cases.

The definition (\ref{deft}) of $\D^0_\omega$ makes the operator ergodic under 2-shifts with respect to the matrix representation (\ref{matrixt}). If $\Sigma$ denotes both the map from  $\Omega \ra \Omega$ such that $(\Sigma \omega)_j=\omega_{j+2}$, and the operator defined  on $\cH_0$ by $\Sigma e_j =e_{j+2},\  \forall j\in \Z, $ we have
\be
T_{\Sigma ^k\omega}=\Sigma^{-k}T_\omega \Sigma^k, \ \forall k\in \Z.
\ee
Following \cite{D1, D2} in making use of independence of the random phases and Borel-Cantelli Lemma, we get
\begin{prop}\label{xxx} Let $l\in 2\N$ and $\theta^{(l)}=(\theta_1, \theta_2, \cdots,\theta_l)\in (\mbox{supp }d\nu)^l\subset\T^l$. Set
$
T_{\theta^{(l)}}:=\D^0_\omega T, \ \mbox{where} \ \omega= (\dots, \theta^{(l)}, \theta^{(l)}, \dots)\in \Omega.
$
Then, 
\be
\cup_{l\in 2\N}\cup_{\theta^{(l)}\in\T^l}\sigma(T_{\theta^{(l)}})\subset \sigma (T_\omega), \ \mbox{almost surely.}
\ee
\end{prop}
\begin{rem}\label{remopt} In particular, if $d\nu(\theta)=d\theta/(2\pi)$, $\cup_{\theta\in [0,2\pi]}e^{i\theta}(\mbox{\em Ran }\lambda_+\cup \mbox{\em Ran }\lambda_-)\subset \sigma(T_\omega)$, where $\lambda_\pm$ are defined in (\ref{evti}). This shows that statements (\ref{87})  and Theorem (\ref{suff}) on the location of $\sigma(T_\omega)$ are optimal, as we argue below.
\end{rem}
Considering Example \ref{ex1}, one checks that when $\xi\ra 0$, condition (\ref{crxi}) holds,  $\lambda_+(0)=\frac{1}{2}\left(\cos(\eta)(1+\sin(\xi))+\sqrt{\cos^2(\eta)(1+\sin(\xi))^2-4\sin(\xi)}\right)>0$ and  the value $r(V)$ given in (\ref{rxi}) becomes arbitrarily close to $\lambda_+(0)$. Also, when $\cos^2(\eta)<4\sin(\xi)/(1+\sin(\xi))^2$ we have $|\lambda_+(0)|=g=\sin (\xi)$. Since $|\lambda_+(0)|\in \sigma(T_\omega)$ almost surely, Proposition \ref{xxx} shows that statement (\ref{87})  and Theorem (\ref{suff}) on the location of $\sigma(T_\omega)$ are optimal.

\section{Special Case $g=0$}
This section is devoted to a more thorough analysis of the case $g=0$ 
\be\label{kg0}
T_\omega=V_\omega P_1 \ \mbox{corresponding to }\ 
\tilde C=
\begin{pmatrix}\alpha & r & \beta \cr
q & 0 & s \cr
\gamma & t & \delta \end{pmatrix}\in U(3). 
\ee

According to Lemmas \ref{symlem} and \ref{easyspect}, $T_\omega=V_\omega P_1$ is far from being unitary, $\ker \ T_\omega=\cH_2$, for all $\omega\in \Omega$, and $\sigma(T_\omega)=\sigma(P_1V_\omega P_1)\cup\{0\}$. More precisely:
\begin{prop}\label{spwp} If $g=0$, we have for all $\omega\in\Omega$
\be
\sigma(T_\omega )\setminus\{0\}\subset \big\{\big| |\alpha|-|\delta| \big|\leq |z|\leq |\alpha|+|\delta|\big\}.
\ee
If $\alpha=0$, resp. $\delta=0$, then $P_1V_\omega P_1|_{\cH_1}$ is unitarily equivalent to $|\delta| S^+$, resp. $|\alpha| S^-$, and 
\bea
&&\sigma(T_\omega)=\max (|\alpha|, |\delta|)\Ss \cup\{0\}\ \mbox{ and }\  \sigma_{p}(T_\omega)=\sigma_{p}(T_\omega^*)=\{0\}. 
\eea
Moreover, \vspace{-.8cm}
\bea
\gamma\neq qt \Leftrightarrow \beta\neq sr \ &\Rightarrow& \ V_\omega \ \mbox{is pure point a.s.} \\
\gamma= qt \Leftrightarrow \beta= sr \ &\Rightarrow& \ V_\omega \ \mbox{is purely ac, } \ \forall \omega\in\Omega. 
\eea
\end{prop}
\begin{exple}{\em
Let us consider an explicit parametrization of a $\tilde C\in O(3)$ of the kind (\ref{kg0})
\be\label{ctino3}
\tilde C(\xi,\eta)=
\begin{pmatrix}\cos(\xi)\sin(\eta) & \cos(\eta) & -\sin(\xi)\sin(\eta) \cr
\sin(\xi) & 0 & \cos(\xi)\cr
-\cos(\xi)\cos(\eta) & \sin(\eta) &  \sin(\xi)\cos(\eta) \end{pmatrix}\in O(3), \ \xi, \eta\in [0,\pi/2],
\ee
where $(\xi,\eta)$ is restricted to $[0,\pi/2]^2$ for simplicity. Then, $|\alpha|+|\delta|<1$ is equivalent to $\sin(\xi+\eta)\neq 1$, i.e. $\xi+\eta\neq \pi/2$, and $\gamma= at$ is equivalent to $\cos(\xi-\eta)=0$, i.e. $(\xi, \eta)=(\pi/2, 0)$, or $(\xi,\eta)=(0,\pi/2)$.
}
\end{exple}

\begin{proof}
Remark \ref{nicexp} implies for $g=0$ that the modulus of the coefficients of the tridiagonal operator $P_1V_{\omega}P_1$ are $|\alpha|$ and $|\delta|$, so Proposition \ref{spw} yields the first statement. We know that $0\in \sigma_p(T_\omega)$. Further assuming that $\alpha\delta=0$, the same remark yields that $P_1V_\omega P_1$ is unitarily  equivalent to a shift and consequently, Lemma \ref{easyspect} yields the spectrum of $T_\omega$. Finally, the eigenvalue equation $T_\omega\ffi=\lambda \ffi$, $\lambda\neq 0$, implies that $\ffi_1=P_1\ffi$ satisfies $P_1V_{\omega}P_1\ffi_1=\lambda \ffi_1$, which cannot hold for a shift. The same argument applies to $T_\omega^*$. Then one checks on the unitary operator (\ref{matv}) that $\gamma = qt $ is equivalent to $\beta=sr $. In turn, this implies that $V_\omega$ is unitarily equivalent to a direct sum of two shifts. In all other cases, $V_\omega$ is pure point almost surely as shown in \cite{JM}. \ep
\end{proof}

From the foregoing we know that  when $g=0$, 
$P_1V_\omega P_1=\D_\eta^{(1)}V_{11}\D_\xi^{(1)}$, where 
\bea
&&V_{11}=\begin{pmatrix}
\ddots & \delta & &  \cr
\alpha & 0 & \delta   &  \cr
 & \alpha & 0 & \delta   \cr
 & &\alpha &  \ddots
\end{pmatrix} \ \simeq \ e^{i(\arg \alpha - \arg \beta)/2}(e^{iy}|\alpha|+e^{-iy}|\delta|), \ \mbox{on }\ L^2(\T), 
\eea
and 
\be
e^{-i(\arg \alpha - \arg \delta)/2}\sigma(V_{11})= E(|\alpha|,|\delta|),
\ee
where $E(|\alpha|,|\delta|)$ denotes the ellipse centered at the origin, with horizontal major axis of length $|\alpha|+|\delta|$ and vertical
minor axis of length $||\alpha|-|\delta||$. When the random phases are iid and uniform, we have a complete description of the spectral properties of $T_\omega$ when $g=0$.
\begin{prop}\label{g=0}
Assume $g=0$ and  $d\nu(\theta)=d\theta/2\pi$. Then, $T_\omega=V_\omega P_1$ satisfies
\be
\sigma(T_\omega)=\{0\}\cup\big\{\big| |\alpha|-|\delta| \big|\leq |z|\leq |\alpha|+|\delta|\big\}, \ \mbox{a.s.}
\ee
When $|\alpha|+|\delta|=1$, the peripheral spectra of the relevant operators coincide with $\Ss$,
\be
\sigma(T_\omega)\cap\Ss=\sigma(P_1V_\omega P_1|_{\cH_1})\cap\Ss=\sigma(V_\omega)=\Ss, \ \mbox{a.s.}
\ee
However, the nature of the peripheral spectra of $T_\omega$ and $V_\omega$ differs   for $\gamma\neq qt$,
\be
\sigma_p(T_\omega)\cap\Ss=\sigma_p(T_\omega^*)\cap\Ss=\emptyset, \ \mbox{whereas} \ \sigma_c(V_\omega)=\emptyset \mbox{ a.s.}
\ee
\end{prop}
\begin{rem} This result shows in a sense that the spectral localization of $V_\omega$ does not carry over to the boundary of the spectrum of $T_\omega=V_\omega P_1$. Note that the original operator $U_\omega(C)$ is purely ac when $g<1$, for all $\omega\in\Omega$.
\end{rem}
\begin{proof}
The first consequence of our assumption on the distribution of the random phases is that  
$
P_1V_\omega P_1=\D_\omega^{(1)}V_{11},
$
where the random phases of the diagonal operator 
$\D_\omega^{(1)}$ are independent and uniformly distributed, see e.g. Lemma 4.1 in \cite{ABJ}. Hence proposition \ref{xxx} with $\mbox{supp } d\nu(\cdot)=2\pi$, together with Proposition \ref{spwp} show that
\be
\big\{\big||\alpha|-|\delta|\big|\leq |z|\leq |\alpha|+|\delta|\big\}=\bigcup_{\theta\in [0,2\pi[} e^{i\theta}E(|\alpha|,|\delta|) = \sigma(\D_\omega^{(1)} V_{11}), \ \mbox{almost surely.}
\ee
When $|\alpha|+|\delta|=1$, the peripheral spectra equals $\Ss$ almost surely  by Lemma \ref{easyspect}.
Finally, the nature of the peripheral spectra stems from Lemmas \ref{kerzero} and \ref{lcnu}. \ep
\end{proof}
\begin{rem} In case $|\alpha|=|\delta|=1/2$, $V_{11}=\Delta_1$, the discrete Laplacian on $\cH_1$. With  $d\nu(\theta)=d\theta/2\pi$,
\be
\sigma(\D_\omega^{(1)} \Delta_1)=\sigma(T_\omega)=\{|z|\leq 1\}, \ \mbox{almost surely},
\ee
where  $\D_\omega^{(1)} \Delta_1$ is a version of the random hopping model of Feinberg and Zee \cite{FZ}.
\end{rem}

\appendix
\section{Proof of Lemmas \ref{form}, \ref{form2}, 
and Proposition \ref{xxx}. }\label{66}

\begin{proof}[of Lemma \ref{form}] The determination of $\partial D(\theta)$ follows from the elimination of the parameter $\tau$ according to (\ref{xtau}) by an explicit computation. 

The relation $D(\theta)\cup B_0(g)\subset \bigcup_{\tau\in \R_+ }B_{\tau}(d_{\tau})\cap B_{g\tau}(gd_{\tau} )$ holds by construction. Let us check that $R_1(\theta)$ belongs to (\ref{symset})
as well.  Let $(x_\tau, y_\tau)=C_\tau(d_\tau)\cap C_{g\tau}(gd_\tau)$.  In order to assess the property  $(x_\tau, y)\in \bigcup_{\tau'\in \R_+ }B_{\tau'}(d_{\tau'})\cap B_{g\tau'}(gd_{\tau'} )$, for some $y\in \R$, we compute for any $\tau'\in \R$,
\bea\label{idset}
(x_\tau-\tau')^2+y^2&=&d_{\tau'}^2+ (y^2-y_\tau^2)+2(\tau'-\tau)(\cos(\theta)-x_\tau)\nonumber\\
(x_\tau-g\tau')^2+y^2&=&gd_{\tau'}^2+(y^2-y_\tau^2)+2g(\tau'-\tau)(g\cos(\theta)-x_\tau).
\eea
Thus, for any $\tau\geq (1+g)/(2g\cos(\theta))$ so that $x_\tau>\cos(\theta)$, and any $y^2\geq y_\tau^2$, we can take $\tau'$ large enough so that $(x_\tau,y)\in B_\tau'(d_{\tau'})\cap B_{g\tau'}(gd_{\tau'})$.

Consider now the reverse inclusion $\bigcup_{\tau\in \R_+ }B_{\tau}(d_{\tau})\cap B_{g\tau}(gd_{\tau} )\subset D(\theta)\cup B_0(g)\cup R_1(\theta)$. By symmetry it is enough to focus on $y  \geq 0$ and $x\leq \cos(\theta)$. 
Using (\ref{idset}) again, we first see that points $(x_\tau, y)\not \in D(\theta)\cup B_0(g)\cup R_1(\theta)$ such that $g\cos(\theta)\leq x_\tau \leq \cos(\theta)$ and $y\geq y_\tau$ cannot belong to $B_{\tau'}(d_{\tau'})\cap B_{g\tau'}(gd_{\tau'} )$, for any $\tau'$. Assume now $(x, y)\not \in D(\theta)\cup B_0(g)\cup R_1(\theta)$ is such that $x\leq g\cos(\theta)$ and $y^2\geq g^2-x^2$. For any $\tau'>0$, the relation
\be
(x-\tau')^2+y^2=gd_{\tau'}^2+(y^2-(g^2-x^2)) +2\tau'(g\cos(\theta)-x)\geq gd_{\tau'}^2
\ee
shows that $(x,y)\not \in \bigcup_{\tau'\in \R_+ }B_{\tau'}(d_{\tau'})\cap B_{g\tau'}(gd_{\tau'} )$, which ends the proof for $\theta<\pi/2$. 

When $\pi/2\leq \theta<\pi$, one first notes that $B_{\tau}(d_{\tau})\cap B_{g\tau}(gd_{\tau} )=B_{g\tau}(gd_{\tau} )$. Then, any $z\in \C$ such that $\Re z>g\cos(\theta)$ is contained in 
$B_{g\tau}(gd_{\tau} )$ provided $\tau>0$ is large enough.

Finally, we prove (\ref{alphaind}) assuming $0< \alpha<\theta<\pi/2$. We first note that if $e^{i\alpha}\tau$ is such that $\Im e^{i\alpha}\tau \geq \sin(\theta)$, {\em i.e.} $\tau\geq \sin(\theta)/\sin(\alpha)$, then any $ z\in B_{e^{i\alpha}\tau}(d_{e^{i\alpha}\tau})$, satisfies $\Re z> \cos(\theta)$, so that $\bigcup_{\tau\geq  \sin(\theta)/\sin(\alpha)}B_{e^{ i\alpha}\tau}(d_{e^{ i\alpha}\tau})\cap B_{ge^{ i\alpha}\tau}(gd_{e^{ i\alpha}\tau} )\subset R_1(\theta)$. For any $\tau<\sin(\theta)/\sin(\alpha)$, the intersection of the line passing by $e^{i\theta}$ and $e^{ i\alpha}\tau$ and the real axis occurs at a point $\tau'>0$ so that $d_{\tau'}=d_{e^{i\alpha}\tau}+|e^{i\alpha}\tau-\tau'|$. Therefore,
if $z\in B_{e^{ i\alpha}\tau}(d_{e^{ i\alpha}\tau})\cap B_{ge^{ i\alpha}\tau}(gd_{e^{ i\alpha}\tau} )$, we have
\bea
&&|z-\tau'|\leq |z-e^{i\alpha}\tau|+|e^{i\alpha}\tau-\tau'|<d_{e^{i\alpha}\tau}+|e^{i\alpha}\tau-\tau'|=d_{\tau'}\nonumber\\
&&|z-g\tau'|\leq |z-ge^{i\alpha}\tau|+g|e^{i\alpha}\tau-\tau'|<gd_{e^{i\alpha}\tau}+g|e^{i\alpha}\tau-\tau'|=gd_{\tau'}
\eea
which shows that $z\in B_{\tau'}(\tau')\cap B_{ge^{\tau'}}(gd_{\tau'} )$ and which ends the proof. A similar argument yields the result for $\pi/2\leq\theta<\pi$.
\ep 
\end{proof}
\begin{proof}[of Lemma \ref{form2}] We consider $0<\theta<\pi/2$ only, the other case being similar. Let $z=\rho e^{i\beta}\in B_0(g)\cup \Delta_g(\theta)$. By symmetry and the foregoing, we can consider $0\leq \beta\leq \pi$ only, and $\rho\geq g$. Thus, it is enough to consider $0\leq \beta < \theta$, and $g\leq \rho <g/\cos(\theta-\beta)$. We need to show that $|\rho e^{i\beta}+|\tau|e^{i\nu}|<|\tau|+g$, for some $\tau\in \R^-$ and some $e^{i\nu}\in \sigma(V)$, which is equivalent to 
\be\label{trian}
2|\tau|(g-\rho\cos(\nu-\beta))>\rho^2-g^2\geq 0.
\ee
Since we have $\cos(\nu-\beta)\leq \cos(\theta-\beta)$, the left hand side of (\ref{trian}) is bounded below by
$2|\tau|(g-\rho\cos(\theta-\beta))$ which is strictly positive, so that (\ref{trian}) holds for $|\tau|$ large enough.  

Conversely, assume $\exists\ |\tau|$ such that  $\forall \ e^{i\nu}\in \sigma(V)$, we have $|z+e^{i\nu}|\tau||<|\tau|+g$. With $z=\rho e^{i\beta}$, 
the geometrical properties recalled above imply that for all $\beta\not\in ]-\theta, \theta[$, $\rho<g$. Otherwise, the inequality is equivalent to
\be\label{sqz}
\rho^2+2\rho|\tau|\cos(\beta-\nu)-g(2|\tau|+g)<0.
\ee
Therefore, denoting by $x_+(\nu)$ the positive root of (\ref{sqz}), we must have for all allowed $\nu$, 
$0\leq \rho \leq x_+(\nu)$, where $\beta \in [-\theta, \theta]$. With $x_+(\nu) \geq x_+(\theta)$, as a consequence of $\cos(\nu-\beta)< \cos(\theta-\beta)$, we must have $0\leq \rho \leq x_+(\theta)$, for $\beta$ fixed. To get the result, one finally checks that $x_+(\theta)<g/\cos(\theta-\beta)$.

Consider now (\ref{reduct}) and fix $\tau\leq 0$. Expression (\ref{symset2}) with $e^{-i\alpha}V$ in place of $V$ and the observation that $\delta_{\tau e^{ i\alpha}}>|\tau |$ implies all circles $C_{e^{i\nu}\tau}(\delta_{\tau e^{ i\alpha}})$ are tangent to $C_0(\delta_{\tau e^{ i\alpha}}-|\tau|)$ yield (\ref{reduct}). Note that $\Gamma_{|\tau|, \delta_{\tau e^{i\alpha}-|\tau|}}(\theta)=B_0(g)$ if $\tau=0$.

It remains to establish (\ref{included}) for $\alpha\geq 0$. We start with a few facts for $|\tau|$ fixed
\bea
&&ge^{\pm i(\theta-\alpha)}\in \partial \Gamma_{|\tau|, \delta_{\tau e^{i\alpha}-|\tau|}}(\theta)\cap C_{\tau e^{\pm i\theta}}(\delta_{\tau e^{i\alpha}})
\eea
The point of of $e^{i\alpha}
\Gamma_{|\tau| , \delta_{\tau e^{ i\alpha}}-|\tau|}(\theta)
$ that is most distant from the origin is 
$e^{i\alpha}\rho_{|\tau|}\in C_{\tau e^{i\theta}}(\delta_{\tau e^{i\alpha}})\cap C_{\tau e^{-i\theta}}(\delta_{\tau e^{i\alpha}})$, where
 \be \rho_{|\tau|}=-|\tau|\cos(\theta)+\sqrt{(g+|\tau|\cos(\alpha))^2+|\tau|^2(\cos^2(\theta)-\cos^2(\alpha))}.
 \ee
 Now, if $\pi/2>\alpha> \theta$, $ \rho_{|\tau|}< g$, so that (\ref{included}) is contained in $B_0(g)$. Thus we assume from now on that $\alpha\leq \theta<\pi/2$. The line tangent to $e^{i\alpha}\partial \Gamma_{|\tau|, \delta_{\tau e^{i\alpha}-|\tau|}}(\theta)$ at $ge^{i\theta}$ has equation 
 \be
 t_{|\tau|}(x)=-(x-g\cos(\theta))\frac{(g\cos(\theta)+|\tau|\cos(\theta+\alpha))}{(g\sin(\theta)+|\tau|\sin(\theta+\alpha))}+g\sin(\theta).
\ee
Note that the tangent to $e^{i\alpha}\partial \Gamma_{|\tau|, \delta_{\tau e^{i\alpha}-|\tau|}}(\theta)$ at $ge^{i(2\alpha-\theta)}$ has slope inferior to $\pi/2$. By convexity, $e^{i\alpha}\Gamma_{|\tau|, \delta_{\tau e^{i\alpha}-|\tau|}}(\theta)\subset \Delta_{|\tau|}$, where $\Delta_{|\tau|}$ is the triangle defined by the intersection point of these tangent lines, $ge^{i\theta}$ and $ge^{i(2\alpha-\theta)}$ union $B_0(\delta_{\tau e^{i\alpha}-|\tau|})$. Since the slope of the line $ t_{|\tau|}$ is strictly increasing with $|\tau|$, we also have $\Delta_{|\tau|}\subset \Delta_{\infty}$, where the latter is set is the triangle is defined by $g\frac{\cos(\alpha)}{\cos(\theta)}e^{i\alpha}$,  $ge^{i\theta}$ and $ge^{i(2\alpha-\theta)}$ union $B_0(\delta_{\tau e^{i\alpha}-|\tau|})$.To prove (\ref{included}), it is enough to show that the line $t_{\infty}$ does not intersects the curve (\ref{cubicurve}) that defines $D(\theta)$ for $x\in ]g\cos(\theta), \min(g\frac{\cos^2(\alpha)}{\cos(\theta)},\cos(\theta)) [$.  With $y(x)>0$  solution to (\ref{cubicurve}), we get
\be
y^2(x)-t_{\infty}^2(x)=\frac{(x-g\cos(\theta))(x^2-x(\cos(\theta)+2g\cos(\theta+\alpha)\cos(\alpha))+g(g+1)\cos^2(\alpha))}{((1+g)\cos(\theta)-x)\sin^2(\theta+\alpha)},
\ee
which has the sign of the second factor in the numerator, call it $p(x)$, for $x\in ]g\cos(\theta),(g+1)\cos(\theta)[$. Moreover, we note that 
\bea
p\left(g\frac{\cos^2(\alpha)}{\cos(\theta)} \right)&=&\frac{g^2\cos^2(\alpha)(\cos^2(\alpha)+\cos^2(\theta)-2\cos(\theta)\cos(\alpha)\cos(\theta+\alpha))}{\cos^2(\theta)}>0.
\eea 
And since $t'_{\infty}(x)=-\frac{\cos(\theta+\alpha)}{\sin(\theta+\alpha)}<\tan(\theta)$, we have $t_\infty(\cos(\theta))<\sin(\theta)=y(\cos(\theta))$, hence $p(\cos(\theta)>0$.
If the discriminant of $p$ is negative, then $p(x)$ has no real roots, $y^2(x)-t_{\infty}^2(x)>0$ and the result holds. Otherwise, denote by $x_- \leq x_+$ these roots such that $x_- x_+=g(g+1)\cos^2(\alpha)>0$. Hence $y^2(x)-t_{\infty}^2(x)$ will be positive on $]g\cos(\theta), \min(g\frac{\cos^2(\alpha)}{\cos(\theta)},\cos(\theta)) [$ if $x_-\leq x_+ < 0$, which happens if and only if  $(\cos(\theta)+2g\cos(\theta+\alpha)\cos(\alpha))<0$. The foregoing yields that neither $g\frac{\cos^2(\alpha)}{\cos(\theta)}$ nor $\cos(\theta)$ lies between the roots. If $g\frac{\cos^2(\alpha)}{\cos(\theta)}\leq \cos(\theta)$, we get
\be
g\frac{\cos^2(\alpha)}{\cos(\theta)}<\cos(\alpha)\sqrt{g(g+1)}\leq\frac12((\cos(\theta)+2g\cos(\theta+\alpha)\cos(\alpha)))\leq x_+,
\ee
and the result follows. If $g\frac{\cos^2(\alpha)}{\cos(\theta)}>\cos(\theta)$, the same largument shows that $\cos(\theta)<x_+$, which ends the proof for $0<\theta<\pi/2$. When $\pi/2\leq \theta<\pi$, the inclusion (\ref{included}) follows directly from (\ref{reduct}) and the simple shape of $B_0(g)\cup R_g(\theta)$.
\ep
\end{proof}
 {

}
\end{document}